\documentclass[10pt,twocolumn,twoside]{IEEEtran}

\usepackage{amsmath} 
\usepackage{amssymb}  
\usepackage{amsthm}
\usepackage{mathtools}
\usepackage{graphicx}
\usepackage[comma,numbers,square,sort&compress]{natbib}
\usepackage{caption,subcaption}
\usepackage{color}
\usepackage{enumerate}
\usepackage{epstopdf}
\usepackage{dsfont}
\usepackage{tabularx,booktabs}

\usepackage{pgfplots}
\pgfplotsset{width=10cm, compat=1.9, legend style={font=\footnotesize}}
\tikzset{%
dot/.style={circle, fill=black, minimum size=4pt, inner sep=0pt, outer sep=-1pt},
hdot/.style={circle, fill=white, minimum size=4pt, inner sep=0pt, outer sep=-1pt},
}
\newlength\fheight
\newlength\fwidth
\usepackage{circuitikz}
\usepgflibrary{arrows}

\usepackage{tikz}
\usetikzlibrary{calc, backgrounds}
\usetikzlibrary{shapes.geometric, arrows}

\tikzstyle{process}  = [rectangle, rounded corners, minimum width=2cm, minimum height=1cm,text centered, text width=2cm, draw=black, fill=red!20]

\tikzstyle{decision} = [rectangle, minimum width=2cm, minimum height=1cm, text centered, text width=2.5cm, draw=black, fill=orange!20]

\tikzstyle{scenario} = [rectangle, rounded corners, minimum width=2cm, minimum height=0.5cm, text centered, text width=3cm, draw=black, fill=pink!20]

\tikzstyle{arrow}= [thick, ->, >=stealth]

\tikzstyle{system} = [rectangle, rounded corners, minimum width=4cm, minimum height=4cm, text centered, text width=4cm, draw=black, fill=red!40,fill opacity = 0.2]

\tikzstyle{controller} = [rectangle,rounded corners,minimum width=1cm,minimum height = 1cm,text centered,text width=1cm,draw=black,fill=green!50,fill opacity = 0.2]

\tikzstyle{detector} = [rectangle,rounded corners,minimum width=1cm,minimum height = 1cm,text centered,text width=1cm,draw=black,fill=yellow!50,fill opacity = 0.2]

\newtheorem{theorem}{Theorem}
\newtheorem{lemma}{Lemma}

\newtheorem{corollary}{Corollary}

\newtheorem{assumption}{Assumption}
\newtheorem{remark}{Remark}

\newcommand{\Ncal}{{\mathcal{N}}}
\newcommand{\Gcal}{{\mathcal{G}}}
\newcommand{\Vcal}{{\mathcal{V}}}
\newcommand{\Ecal}{{\mathcal{E}}}
\newcommand{\Wcal}{{\mathcal{W}}}
\newcommand{\Acal}{{\mathcal{A}}}
\newcommand{\Lcal}{{\mathcal{L}}}
\newcommand{\Dcal}{{\mathcal{D}}}
\newcommand{\Bcal}{{\mathcal{B}}}
\newcommand{\Scal}{{\mathcal{S}}}
\newcommand{\Ccal}{{\mathcal{C}}}
\newcommand{\RR}{{\mathbb{R}}}
\newcommand{\CC}{{\mathbb{C}}}
\newcommand{\KK}{{\mathbb{K}}}
\newcommand{\I}{{\mathbf{I}}}
\newcommand{\1}{{\mathbf{1}}}
\newcommand{\0}{{\mathbf{0}}}

\newcommand{\commentout}[1]{}

\title{
    On Consensusability of Linear Interconnected Multi-Agent Systems and Simultaneous Stabilization
}

\author{Mustafa S. Turan, Liang Xu, Giancarlo Ferrari-Trecate
    \thanks{Authors are with the Institute of Mechanical Engineering (IGM),
    	EPFL, Switzerland. Email: {\tt\small  \{mustafa.turan, liang.xu, giancarlo.ferraritrecate\}@epfl.ch}}
    \thanks{This work has been supported by the Swiss National Science Foundation under the COFLEX project (grant number 200021\_169906) and the National Centre of Competence in Research (NCCR) in Dependable and Ubiquitous Automation.}
}

\begin{document}
\maketitle

\begin{abstract}
  Consensusability of multi-agent systems (MASs) certifies the existence of a distributed controller capable of driving the states of each subsystem to a consensus value.
  We study the consensusability of linear interconnected MASs (LIMASs) where, as in several real-world applications, subsystems are physically coupled.
  We show that consensusability is related to the simultaneous stabilizability of multiple LTI systems, and present a novel sufficient condition in form of a linear program for verifying this property.
  We also derive several necessary and sufficient consensusability conditions for LIMASs in terms of parameters of the subsystem matrices and the eigenvalues of the physical and communication graph Laplacians.
  The results show that weak physical couplings among subsystems and densely-connected physical and communication graphs are favorable for consensusability.
  Finally, we validate our results through simulations of networks of supercapacitors and DC microgrids.
\end{abstract}\vspace{-0.3cm}

\section{Introduction}\label{sec:Intro}

Consensus theory for multi-agent systems (MASs) finds applications in several fields, ranging from cooperative robotics to power systems~\cite{RezaOlfatiSaber2007ProcIEEE}.
Consensusability of MASs refers to the existence of a distributed protocol, using only locally available information, from a predefined class such that the MAS can achieve consensus. As such, it is a binary property.
For linear MASs, consensusability has been extensively studied in the past decades.
For example, the authors in~\cite{MaCuiqin2010TAC} show that a continuous-time linear MAS can reach consensus if the dynamics of each agent is controllable and the communication topology is connected.
The work~\cite{you2011} shows that, for discrete-time linear MASs, consensusability is guaranteed if the unstable eigenvalues of the agent state matrix verifies certain conditions related to the eigenvalues of the communication graph Laplacian.
For consensus of MASs with switching topologies, \cite{WangXingping2018TAC} shows that if the Lyapunov exponent of agent dynamics is less than a suitably defined synchronizability exponent of the switching topology, the MAS can achieve consensus.
The authors of~\cite{XuLiang2018Automatica} and \cite{XuLiang2020TAC} study consensus in presence of communication channels affected by fading and packet dropouts, and show that the consensusability condition is closely related to
the statistics of the noisy communication channel, the eigenvalues of the communication graph Laplacian, and the instability degree of the agent dynamics. 
The works~\cite{Zhengjianying2019TAC} and \cite{XuJuanjuan2019TAC} consider MASs affected by communication delays and provide upper bounds on delays for guaranteeing consensusability.

The above research works assume that agents in MASs are coupled through a cyber layer only (i.e., a communication network) and not through physical interactions. This is not the case in several real-world applications such as power networks and microgrids, where nodes are physically interconnected and, consequently, their dynamics are coupled~\cite{guerrero2010}.
Other examples of interconnected MASs include flow networks, production systems, and traffic networks~\cite{chen1991discrete,papageorgiou1990dynamic}.
This raises the issue of studying how physical couplings affect consensusability.

This problem has been considered in~\cite{wang2015hierarchical,cheng2015leader,chen2015consensus}.
The authors of~\cite{wang2015hierarchical} focus on consensus of multiple linear systems with uncertain subsystem interconnections for tracking a reference.
They propose a distributed adaptive controller based on hierarchical decomposition and prove that the consensus error converges to a compact set if physical interconnections are sufficiently weak. 
Leader-follower tracking problems for linear interconnected MASs (LIMASs) are considered in~\cite{cheng2015leader}. 
Interactions between systems are treated as dynamic uncertainties and are described in terms of integral quadratic constraints.
Two methods are proposed to design consensus-like tracking protocols. 
Sufficient conditions to guarantee that the system tracks the leader are obtained in terms of the feasibility of linear matrix inequalities (LMIs).
The authors of~\cite{chen2015consensus} investigate the state consensus problem of a general LIMAS.
They propose a linear consensus protocol and derive a sufficient and necessary criterion to guarantee convergence to consensus, which is expressed in terms of the Hurwitz stability of a matrix constructed from the parameters of the agents and the protocols.
However, the aforementioned works lack a quantitative analysis of the relations of this property with physical interconnection and communication graphs.

In this paper, we study the consensusability of LIMASs equipped with linear distributed controllers, while providing analytic characterizations on how the physical and cyber couplings impact on it. 
In particular, we consider \textit{homogeneous} LIMASs, whose subsystems have identical dynamics and control gains (see Section~\ref{sec:ProblemFormulation}). Moreover, we direct our attention to LIMASs consisting of single-input subsystems interconnected via physical coupling with a Laplacian structure. 
The contributions of this paper to the existing literature are given as follows. 
First, we show that the consensusability problem for LIMASs is related to a simultaneous stabilizability problem.
Second, we present a linear-programming-based sufficient condition for verifying the simultaneous stabilizability of multiple LTI systems, which provides a simple alternative to existing methods relying on convex programming. 
Third, we present several sufficient conditions, as well as a necessary condition, for the consensusability of LIMASs.
Our results illustrate how consensusability is influenced by physical and communication coupling among subsystems.

A preliminary version of this work has been presented in~\cite{turan2020consensusability}. With respect to it, this paper has several differences.
In addition to providing proofs for Theorems~\ref{prop:consensusability_scalar}-\ref{prop:consensusability_necessary}, new tests for the simultaneous stabilizability of multiple LTI systems are developed. Based on these results, this paper also
provides new sufficient conditions for consensusability. Finally, simulation results have been completely revisited and include applications to microgrids. 

The paper is organized as follows.
In Section~\ref{sec:ProblemFormulation}, we introduce LIMASs and provide the problem formulation.
In Section~\ref{sec:SimultaneousStabilization}, we discuss the simultaneous stabilization problem of multiple LTI systems.
Consensusability analysis of LIMASs is presented in Section~\ref{sec:Consensusability}.
Simulation results are given in Section~\ref{sec:SimulationResults} and concluding remarks are presented in Section~\ref{sec:ConclusionAndFuturePerspectives}.\vspace{-0.2cm}

\subsection{Notation}
The operator $|\cdot|$ applied to a set determines its cardinality, while used with matrices or vectors it defines their component-wise absolute value.
For a symmetric matrix $A\in\RR^{n\times n}$, $A\succ0$ ($A\succeq 0$) denotes that it is positive (semi-) definite. When used with vectors, inequalities are meant component-wise. 
$\rho(.)$ represents the spectral radius of its argument.
The symbol $\Bcal_1$ denotes the unit disk in the complex plane. A polynomial is called Schur if all of its roots are in $\Bcal_1$.
The symbol $\otimes$ represents the Kronecker product. With a slight abuse of notation, it is similarly defined for subspaces, i.e., for two subspaces $\mathbb{X}$ and $\mathbb{Y}$, $\mathbb{X}\otimes\mathbb{Y} = \{x\otimes y| x\in \mathbb{X}, y\in \mathbb{Y}\}$.
$\1_n \in \RR^n$ and $\0_n \in \RR^n$ represent column vectors with all elements equal to one and zero, respectively. $\0_{N\times n}\in \RR^{N\times n}$ and $\I_n\in \RR^{n\times n}$ denote a matrix that consists of all zero elements and an identity matrix, respectively. $\mathbb{H}^1$ denotes the $(N-1)$-dimensional subspace of $\RR^N$ comprising vectors with zero average, i.e., $\mathbb{H}^1 = \{v\in\RR^N| \1_N^\top v = 0\}$. $\mathbb{H}^1_\perp$ is the $1$-dimensional subspace orthogonal to $\mathbb{H}^1$, composed of $N$-dimensional vectors of identical elements, i.e., $\mathbb{H}^1_\perp \perp \mathbb{H}^1$ and $\mathbb{H}^1_\perp = \{a\1_N| a\in \RR\}$. Then, $\mathbb{H}^1\oplus \mathbb{H}^1_\perp = \RR^N$, where $\oplus$ denotes the direct subspace sum. The set $\{-1,1\}^{m\times n}$ with cardinality $2^{mn}$ consists of all the different $m\times n$ matrices comprising elements $-1$ and $1$. $V=\text{col}\{v_1,\dots,v_N\}$ denotes a matrix composed of column concatenation of vectors $v_i,$ $i=\{1,\dots,N\}$. \vspace{-0.2cm}

\subsection{Preliminaries on algebraic graph theory}

An undirected weighted graph of $N$ nodes is defined as $\Gcal = \left(\Vcal, \Wcal, \Ecal\right)$, where $\Vcal = \{1,2,\dots,N\}$ is the set of nodes, $\Ecal \subseteq \Vcal \times \Vcal$ is the set of edges, and $\Wcal$ is a diagonal matrix with the weight of the corresponding edge in each diagonal entry.
    The set of neighbors of node $i$ is defined as $\Ncal_i = \{j|(i,j)\in\Ecal\}$.
    A path $p_{ij}$ is an ordered sequence of consecutive edges such that every edge in the sequence is in $\Ecal$, the first edge starts from node $i$, and the last edge ends in node $j$.
    The adjacency matrix of $\Gcal$ is defined as $\Acal = \left[a_{ij}\right]_{N\times N}$, where $a_{ij} = 0$ if $\left(i,j\right)\notin \Ecal$ or $i=j$ and $a_{ij} >0$ is the positive weight of edge $\left(i,j\right) \in \Ecal$.
    The degree of a node $i$ is defined as $d_i = \sum_{j\in\Ncal_i}a_{ij}$ along with the degree matrix $\Dcal = \mathrm{diag}\left(d_1,d_2,\dots,d_N\right)$.
    The Laplacian matrix of $\Gcal$ is given by $\Lcal = \Dcal - \Acal$.
    An undirected graph is connected if there exists a path from every node $i$ to every other node. \vspace{-0.15cm}

\section{Problem Formulation}\label{sec:ProblemFormulation}

    The interaction among agents in a LIMAS is described by two undirected graphs with a common set of nodes $\Vcal = \{1,\dots,N\}$ associated with subsystems: a physical graph $\Gcal_p = \left(\Vcal,\Wcal_p,\Ecal_p\right)$
    representing the physical interconnection among subsystems and a cyber graph $\Gcal_c = \left(\Vcal,\Wcal_c,\Ecal_c\right)$ representing the communication network. 
    Figure~\ref{fig:LIMAS_struct} shows an example LIMAS.
In the scope of this work, we assume that $\Gcal_c$ is connected; however, the physical graph $\mathcal{G}_p$ is allowed to be disconnected. For a subsystem $i$, the set of its neighbors $\Gcal_p$ and $\Gcal_c$ are denoted by $\Ncal_i^p$ and $\Ncal_i^c$, respectively.

    We consider a homogeneous LIMAS with subsystem dynamics described by
\begin{equation}\label{eq:Si_dyn}
x_i^+ = Ax_i + A_p\sum_{j\in\Ncal_i^p}a_{ij}(x_j-x_i) + Bu_i, \quad i=1, \ldots, N,
  \end{equation}
where $x_i \in \RR^n$ are the states, $u_i \in \RR$ are the scalar control inputs, $a_{ij}=a_{ji} \in \RR_{>0}$ are the symmetric physical interconnection weights, and $A_p \in \RR^{n\times n}$ is a matrix determining the physical coupling.
We are interested in distributed controllers given by
\begin{equation}\label{eq:Si_contr}
u_i = K\sum_{j\in\Ncal_i^c}b_{ij}\left(x_i-x_j\right), \quad i=1, \ldots, N,
\end{equation}
where $K\in\RR^{1\times n}$ is the control gain common to all subsystems and $b_{ij}=b_{ji}\in \RR_{>0}$ denote the symmetric communication weights in the cyber graph $\Gcal_c$. Note that the control gain $K$ is a design parameter while $b_{ij}$, $\Lcal_c$ are assumed to be given.

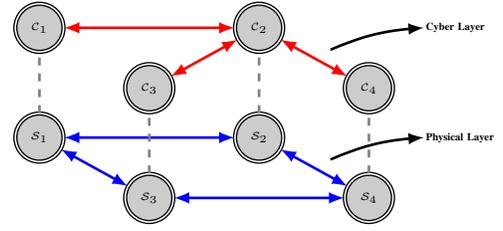
\begin{figure}\hspace{1.5cm}
	\centering
	{\tiny
		\ctikzset{bipoles/length=0.8cm}
		\tikzstyle{every node}=[minimum size=.8cm, inner sep=.8, line width=0.5pt, scale=0.8]
		\begin{circuitikz}[american currents, scale=0.73, line width=1pt]
			\draw (1,1-2.0) node(s1) [circle, draw=black, double,
			fill=black!20] {$\Scal_1$};
			\draw (7,-.1-2.0) node(s4)  [circle, draw=black, double,fill=black!20] {$\Scal_4$};
			\draw (3,-.1-2.0) node(s3)  [circle, draw=black, double, fill=black!20] {$\Scal_3$};
			\draw (5,1-2.0) node(s2)  [circle, draw=black, double,  fill=black!20] {$\Scal_2$};

			\draw[latex-latex, blue] (s1) to (s2);
			\draw[latex-latex, blue] (s2)  to (s4);
			\draw[latex-latex, blue] (s4) to (s3);
			\draw[latex-latex, blue] (s1) to (s3);

			\draw (1,1) node(c1) [circle, draw=black, double,
			fill=black!20] {$\Ccal_1$};
			\draw (7,-.1) node(c4)  [circle, draw=black, double,fill=black!20] {$\Ccal_4$};
			\draw (3,-.1) node(c3)  [circle, draw=black, double, fill=black!20] {$\Ccal_3$};
			\draw (5,1) node(c2)  [circle, draw=black, double,  fill=black!20] {$\Ccal_2$};
			
			\draw[dashed, gray] (s1) to(c1);
			\draw[dashed, gray] (s2) to (c2);
			\draw[dashed, gray] (s3) to (c3);
			\draw[dashed, gray] (s4) to (c4);
			
			\draw[latex-, black] (8,1-2.0) to [bend right=10] (6.3,0.6-2.0) ;
			\draw[latex-, black] (8,1) to [bend right=10] (6.3,0.6) ;
			
			\node[anchor=west,text width=1.5cm] (note1) at (8,1-2.0)
			{\textbf{Physical Layer}};
			
			\node[anchor=west,text width=1.5cm] (note2) at (8,1)
			{\textbf{Cyber Layer}};
			
			\draw[latex-latex, red] (c3) to (c2);
			\draw[latex-latex, red] (c2) to (c1);
			\draw[latex-latex, red] (c4) to (c2);
	\end{circuitikz}}
	\caption{Illustration of a LIMAS. Blue arrows represent physical couplings among subsystems $\Scal_i$, dashed gray lines indicate connections between each subsystem and its corresponding controller, and red arrows represent communication channels among controllers $\Ccal_i$.
	}\vspace{-0.45cm}
	\label{fig:LIMAS_struct}
\end{figure}

    \begin{remark}
      Typical examples of LIMASs that can be modeled by~\eqref{eq:Si_dyn}, \eqref{eq:Si_contr} are DC microgrids, where distributed generation units are physically coupled through electric lines and communication networks are used for obtaining global coordinated behaviors, such as current sharing and voltage balancing~\cite{tucci2018}. A detailed example is provided in Section~\ref{sec:SimulationResults}. 
      We note that, in the field of consensus for MASs, it is common to assume identical subsystem dynamics as in~\eqref{eq:Si_dyn}. Such an assumption is reasonable in several application scenarios, where the hardware of subsystems is standardized for efficient serial production. For instance, individual converters in a DC microgrid can be chosen as identical, resulting in a homogeneous LIMAS. Moreover, using a single static feedback gain for all subsystems alleviates the burden of designing different consensus gains for each subsystem, which is critical in deploying control architectures for large-scale MASs.
  \end{remark}

By combining \eqref{eq:Si_dyn} and \eqref{eq:Si_contr}, the overall dynamics of the LIMAS can be compactly written as:
\begin{equation}\label{eq:LIMAS_dyn}
x^+ = \left(\I_N \otimes A - \Lcal_p\otimes A_p + \Lcal_c\otimes BK\right)x,
\end{equation}
where $x = \left[x_1^{\top},\dots,x_N^{\top}\right]^{\top} \in \RR^{Nn}$ is the cumulative state whereas $\Lcal_p$ and $\Lcal_c$ are the Laplacian matrices of $\Gcal_p$ and $\Gcal_c$, respectively. Note that the structure of physical interconnections in~\eqref{eq:Si_dyn} gives the term $\Lcal_p\otimes A_p$ in~\eqref{eq:LIMAS_dyn} involving the Laplacian matrix $\Lcal_p$. For this reason, the coupling is termed \textit{Laplacian}. Physical interconnections hamper the use of existing methods for analyzing consensusability of MASs. We define the consensusability problem in LIMASs as follows.

\noindent\textbf{Problem 1}:
\textit{Given the LIMAS \eqref{eq:LIMAS_dyn}, provide conditions for the existence of a static feedback gain $K \in \RR^{1\times n}$ such that the states of all subsystems converge to a global consensus vector, i.e., 
	\begin{equation}\label{eq:consensusdefinition}
	\lim\limits_{t\rightarrow\infty}|x_i(t)-\bar{v}| = \0_n, \quad \forall i \in \Vcal,
	\end{equation}
	for some $\bar{v} \in \RR^{n}$.}

To study this problem, we define the average state \vspace{-0.1cm}
\begin{equation*}
	\bar{x} \triangleq \frac{1}{N} \sum_{i=1}^Nx_i = \frac{1}{N}\left(\1_N^{\top}\otimes\I_n\right)x,
\end{equation*}
 and the deviation from $\bar{x}$
\begin{equation}\label{eq:delta_def}
\begin{split}
\delta &\triangleq \left[x_1^{\top}-\bar{x}^{\top},\dots,x_N^{\top}-\bar{x}^{\top}\right]^{\top} = x-\1_N\otimes\bar{x} \\
&= \left(\left(\I_N-\frac{1}{N}\1_N\1_N^{\top}\right)\otimes\I_n\right)x
\end{split}.
\end{equation}

The dynamics of $\delta$ can be derived from \eqref{eq:LIMAS_dyn} and \eqref{eq:delta_def} as\vspace{-0.15cm}

\begin{equation}\label{eq:delta_dyn}
\delta^+
= \left(\I_N \otimes A - \Lcal_p\otimes A_p + \Lcal_c\otimes BK\right)\delta.
\end{equation}
The consensusability of the LIMAS is, therefore, equivalent to the stabilizability of \eqref{eq:delta_dyn}.

Consider the following properties:
	\begin{enumerate}
		\item the columns of $\1_N\otimes\I_n/\sqrt{N}$ span the subspace $\mathbb{H}^1_\perp \otimes \RR^n$
		\item $(\mathbb{H}^1\otimes \RR^n) \perp (\mathbb{H}^1_\perp\otimes \RR^n)$
		\item $(\mathbb{H}^1\otimes \RR^n) \oplus (\mathbb{H}^1_\perp\otimes \RR^n) = \RR^{Nn}$
	\end{enumerate}
	They imply that there exists a unitary matrix $\Phi = \left[\1_N/\sqrt{N},\phi_2,\dots,\phi_N\right]$ where $\{\phi_2,\dots,\phi_N\}$ form a basis for $\mathbb{H}^1$, such that, by defining $\tilde{\delta} = \left[\tilde{\delta}_1^{\top},\dots,\tilde{\delta}_N^{\top}\right]^{\top} \triangleq \left(\Phi^{\top}\otimes\I_n\right) \delta$, we have \vspace{-0.2cm}
	\begin{equation}\label{eq:deltatilde_dyn_vec}
	\begin{split}
	\tilde{\delta}^+ = \left(\I_N \otimes A - \begin{bmatrix}
	\0_{n\times n} & \0_{n\times (N-1)n} \\
	\0_{(N-1)n\times n} & \tilde{\Lcal}_p
	\end{bmatrix}\otimes A_p \right.& \\
	\left. + \begin{bmatrix}
	\0_{n\times n} & \0_{n\times (N-1)n} \\
	\0_{(N-1)n\times n} & \tilde{\Lcal}_c
	\end{bmatrix}\otimes BK\right)\tilde{\delta}.&
	\end{split}
	\end{equation}
	$\tilde{\Lcal}_p$ is a positive semidefinite matrix with eigenvalues $\{\lambda_{p,2},\dots,\lambda_{p,N}\}$. Similarly, the positive definite matrix $\tilde{\Lcal}_c$ has eigenvalues $\{\lambda_{c,2},\dots,\lambda_{c,N}\}$.
	The transformation in~\eqref{eq:deltatilde_dyn_vec} decomposes the dynamics of $\delta$ into two noninteracting parts: $\tilde{\delta}_1$ and $[\tilde{\delta}_2^\top,\dots,\tilde{\delta}_N^\top]^\top$ representing the evolution of $\delta$ in $\mathbb{H}^1_\perp \otimes \RR^n$ and in $\mathbb{H}^1 \otimes \RR^n$, respectively. 
	Furthermore, we can show that $\tilde{\delta}_1=\0_n$ by definition. Therefore, $\delta(t)$ asymptotically converges to zero if and only if the dynamics of $[\tilde{\delta}_2^\top,\dots,\tilde{\delta}_N^\top]^\top$ is stable, i.e., the matrix
	        \begin{align}
	          \label{eq:reduced-closed-loop-matrix}
	        \I_{N-1}\otimes A - \tilde{\Lcal}_p\otimes A_p + \tilde{\Lcal}_c\otimes BK 
	        \end{align}
	      is Schur stable.
As will be shown in Section~\ref{sec:Consensusability}, the stability of~\eqref{eq:reduced-closed-loop-matrix} is related to the problem of simultaneous stabilization of a group of low-dimensional systems.
Therefore, in the next section, we first make a detour and provide novel conditions for simultaneous stabilization that, besides being useful for analyzing consensusability, also have an independent value.\vspace{-0.3cm}

\section{A Simultaneous Stabilization Test Based on Linear Programming}\label{sec:SimultaneousStabilization}

The simultaneous stabilization problem has attracted the interest of several researchers, especially in the area of robust control~\cite{cao1998static}.
Despite a number of results for linear systems~\cite{kabamba1991,vidyasagar1982,blondel1994}, prior work has shown that providing algebraic conditions for simultaneous stabilization of more than three systems is a difficult problem~\cite{blondel1993}.
Extensive effort, therefore, has been put in developing numerical tests~\cite{cao1998static}, mainly relying on convex and non-convex programming (CP and NCP, respectively)~\cite{henrion1999rank, boyd1993control, howitt1991simultaneous,ackermann1980parameter}.
In this section, we study the simultaneous stabilization of multiple LTI systems via linear static feedback.
Existing criteria include sufficient and necessary conditions in terms of NCP~\cite{howitt1991simultaneous,ackermann1980parameter} and LMI-based sufficient conditions~\cite{boyd1993control}.
In this section, we present a sufficient condition in terms of linear programming (LP), which require less computational resources than their CP and NCP counterparts~\cite{boyd2004convex}. 

We consider $M$ single-input LTI systems described by
\begin{equation}\label{eq:sim_stabilization_dyn}
        \check{x}_l^+ = A_l\check{x}_l+B_l\check{u}_l, \quad l =1,\dots,M,
\end{equation}
where $\check{x}_l \in \RR^m$ are system states and $\check{u}_l\in\RR$ are control inputs. In the sequel, we assume that all pairs $(A_l, B_l)$ are controllable.
A first result is stated in Lemma~\ref{lem:SimultaneousStabilization}, which relies on a conservative parameterization of stabilizing controllers based on Ackermann's formula~\cite{shinners1998modern}. Note that controller parameterization using Ackermann's formula has been exploited for simultaneous stabilization also in~\cite{ackermann1980parameter,howitt1991simultaneous}; however, these works provide necessary and sufficient conditions in terms of NCPs, which can not be efficiently checked for very large $m$ and $M$. Before presenting the lemma, we introduce the following definitions. For $l=1, \ldots, M$, $\mathcal{M}_l=[B_l,A_lB_l,\dots,A_l^{m-1}B_l]$ is the controllability matrix of the pair $(A_l,B_l)$, and the last row of its inverse is denoted as $r_{l,m}=[0\enskip\dots\enskip 0 \enskip 1]\mathcal{M}_l^{-1}$. Accordingly, we define
        \begin{equation}\label{eq:V_l}
          \begin{split}
          V_l &= \text{col}\{-r_{l,m}\I_m, -r_{l,m}A_l,\dots,-r_{l,m}A_l^{m-1}\} \\
        v_l& = -r_{l,m}A_l^m
        \end{split}.
      \end{equation}
     We also define $\Gamma=\text{col}\{\Gamma_{r,1},\Gamma_{r,2},\dots,\Gamma_{r,2^m}\}\in\RR^{2^m\times m}$, 
        where $\Gamma_{r,j}$ are the different vectors in the set $\{-1,1\}^{1\times m}$.

\begin{lemma}\label{lem:SimultaneousStabilization}
For the group of controllable systems $(A_l,B_l)$, $l=1, \ldots,M$, there exists $K$ such that $A_l+B_lK$ is Schur stable for every $l$ if there exists a vector $c = [c_1^\top \enskip \dots \enskip c_M^\top]^\top \in \RR^{Mm}$ satisfying
        \begin{equation}\label{eq:LP_suff_prop}
                Vc= v, \quad Hc\leq h,
        \end{equation}
where \vspace{-0.2cm}
                \begin{equation}\label{eq:V_H_definitions}
        \begin{split}
        V = &\begin{bmatrix}
        V_1^\top & -V_2^\top & \0_{m\times m} & \dots & \0_{m\times m} \\
        \0_{m\times m} & V_2^\top & -V_3^\top &  \ddots & \0_{m\times m} \\
        \vdots  & \ddots  & \ddots & \ddots & \vdots \\
        \0_{m\times m}  & \dots & \0_{m\times m}  & V_{M-1}^\top &-V_{M}^\top
        \end{bmatrix}\\
        &\quad \enskip   v = [v_2-v_1,\dots,	v_M-v_{M-1}]^\top \\
         & \quad \quad  H = \I_{M}\otimes\Gamma\quad \quad  h = \1_{M2^m}
        \end{split}.
        \end{equation}
 Moreover, if such a $c$ exists, the matrices $c_l^\top V_l+v_l$ are identical and any of them provides a simultaneously stabilizing gain $K$.
\end{lemma}

\begin{proof}
  We first describe the controller parametrization for the pair $(A_l,B_l)$ and later show how to design a simultaneously stabilizing controller for all pairs $(A_l, B_l)$ based on this parameterization.

In view of Ackermann's formula, we know that, for the controllable pair $(A_l,B_l)$ and a vector collecting desired closed-loop eigenvalues $\lambda^D_l\triangleq [\lambda^D_{l,1},\dots,\lambda^D_{l,m}]^\top\in \CC^m$, the state-feedback controller $K_l= -r_{l,m}p^D_{l}(A_l)\in \RR^{1\times m}$ assigns the eigenvalues of $A_l+B_lK_l$ to the elements of $\lambda^D_l$, 
where the desired characteristic polynomial $p^D_{l}(A_l)$ is written as
        \begin{equation}\label{eq:char_polyn}
        \begin{split}
        p^D_{l}(A_l) &= \left(A_l-\lambda^D_{l,1}\I_m\right)\dots\left(A_l-\lambda^D_{l,m}\I_m\right) \\
        &\triangleq A_l^m + c_{l,m-1}A_l^{m-1} + \dots + c_{l,1}A_l + c_{l,0}\I_m
        \end{split}.
      \end{equation}
Defining $c_l\triangleq [c_{l,0} \enskip \dots \enskip c_{l,m-1}]^\top$, one sees that $c_l = g(\lambda^D_l)$ 
is a polynomial function of order $m$. Therefore, a set $\KK_l$ of stabilizing controllers $K_l$ for $(A_l,B_l)$ can be parameterized by $\lambda_l^D$ as $\KK_l = \{K_l=-r_{l,m}p^D_{l}(A_l)|\lambda^D_l\in \Lambda^D\}$,
where 
\begin{equation*}
\begin{split}
\Lambda^D = \{\lambda^D\in\CC^m| \lambda^D_{j} \in \Bcal_1 \enskip \forall j\in\{1,\dots,m\} \enskip \text{and}& \\
\lambda^D_{j} \text{ are real or in conjugate pairs}& \}
\end{split}.
\end{equation*}

    Since the set $\Lambda^D$ has a complex geometry, the computation of the set $\KK_l$ is convoluted. To circumvent this problem, we leverage a classic result on Schur stable polynomials~\cite{jury1964theory}: 
    $\lambda^D_l\in\Lambda^D$ if \vspace{-0.2cm}
        \begin{equation}\label{eq:dominant_condition}
        \sum_{j=0}^{m-1}|c_{l,j}|<1,
        \end{equation}\vspace{-0.4cm}
        
       \noindent  which can equivalently be written as $\Gamma c_l\leq \1_{2^m}$. Further noting that the polynomial $p^D_l(A_l)$ is an affine function of $c_l$ and $K_l\in\KK_l$ is a linear function of $p_l^D(A_l)$, we can define a new set $\KK_l^s \subseteq \KK_l$ of stabilizing controllers for system $l$ in terms of $c_l$ as $\KK_l^s = \{K_l=c_l^\top V_l + v_l|\Gamma c_l\leq \1_{2^m}\}$,
        where $V_l$ and $v_l$ are defined in~\eqref{eq:V_l}.

        Based on the above results, we know that if
        $\bigcap_{l=1}^M\KK_l^s \neq \emptyset,$
 there exists $K\in\bigcap_{l=1}^M\KK_l^s$ simultaneously stabilizing all $(A_l,B_l)$ pairs.
The above condition is equivalent to the existence of vectors $c_l \in \RR^m$ for $l\in\{1,\dots,M\}$ such that $c_l^\top V_l+v_l=c_{l+1}^\top V_{l+1}+v_{l+1}, \enskip \forall l\in\{1,\dots,M-1\}$, which yields the feasibility condition given by \eqref{eq:LP_suff_prop}. The proof of the second part is straightforward as the feasibility of \eqref{eq:LP_suff_prop} means that stabilizing control gains $K_l=C_l^\top V_l + v_l$, $l=1,\dots,M$ are identical. Therefore, it suffices to pick one.
\end{proof}\vspace{-0.3cm}

\begin{remark}\label{rem:LP_conservativity}
	In Lemma~\ref{lem:SimultaneousStabilization}, the only source of conservativity is the use of the condition in~\eqref{eq:dominant_condition} for Schur stability of polynomials, which is only sufficient for $m>1$, and gets more and more conservative as the system order $m$ increases.
	Indeed,~\eqref{eq:dominant_condition} is the main novelty of the proposed method, compared with the necessary and sufficient conditions in~\cite{ackermann1980parameter,howitt1991simultaneous}. Therefore, for scalar systems, i.e., $m=1$, the results in Lemma~\ref{lem:SimultaneousStabilization} are necessary and sufficient.
The conservativity of this lemma, however, does not increase with the number of systems $M$. 
\end{remark}

We observe that the structure of the equality constraints in~\eqref{eq:LP_suff_prop} can be further exploited to simplify the redundancies in the proposed LP. Therefore, in the following, we show that the result in Lemma~\ref{lem:SimultaneousStabilization} can be equivalently cast into a simpler LP with a smaller number of decision variables.

\begin{theorem}\label{thm:SimultaneousStabilization}
For the group of controllable systems $(A_l,B_l)$, $l=1,\dots,M$, the following hold:
\begin{enumerate}
        \item \label{itm:Vl_inv} The matrices $V_l$ in~\eqref{eq:V_l} are invertible
        \item \label{itm:V_fullrank} The matrix $V$ in~\eqref{eq:V_H_definitions} has full row rank and its right inverse $V^\dagger$ is given by
        \begin{equation}\label{eq:Vdagger}
                V^\dagger = \begin{bmatrix}
                \left(V_1^\top\right)^{-1} & \left(V_1^\top\right)^{-1} & \dots & \left(V_1^\top\right)^{-1} \\
                \0_{m\times m} & \left(V_2^\top\right)^{-1} &  \dots & \left(V_2^\top\right)^{-1} \\
                \vdots & \ddots & \ddots & \vdots \\
                \0_{m\times m} & \dots & \0_{m\times m} & \left(V_{M-1}^\top\right)^{-1} \\
                \0_{m\times m} & \0_{m\times m} & \0_{m\times m} & \0_{m\times m}
                \end{bmatrix}.
        \end{equation}
        Moreover, the null space of $V$ is spanned by columns of the matrix
        \begin{equation}\label{eq:Psi}
                \Psi = \begin{bmatrix}
                V_1^{-1} &V_2^{-1} & \dots &V_M^{-1}
                \end{bmatrix}^\top.
        \end{equation}
        \item \label{itm:LP_reduced} Pairs $(A_l,B_l)$ are simultaneously stabilizable by a common gain $K$ if there exists a vector $w \in \RR^{m}$ such that
        \begin{equation}\label{eq:LP_suff_prop_red}
                H\Psi w\leq h-HV^\dagger v.
        \end{equation}
Furthermore, if such a $ w$ exists, the feedback gain
        \begin{equation}\label{eq:K_sim_stab}
                K=v_M+ w^\top
        \end{equation}
        stabilizes $(A_l,B_l)$ for every $l\in\{1,\dots,M\}$.
\end{enumerate}
\end{theorem}

\begin{proof}
\begin{enumerate}
	\item We prove that the matrix $V_l$ is invertible if the pair $(A_l,B_l)$ is controllable, by showing that the rows of $V_l$ are linearly independent.
	For $m=1$ it is trivial. For $m>1$, we denote the rows of the inverse of $\mathcal{M}_l$ as $$\mathcal{M}_l^{-1} = \begin{bmatrix}
	r_{l,1}^\top & \dots & r_{l,m}^\top
	\end{bmatrix}^\top.$$
	From the last row of the equality $\mathcal{M}_l^{-1}\mathcal{M}_l=\I_m$, one has that $r_{l,m}B_l = r_{l,m}A_lB_l = \dots = r_{l,m}A_l^{m-2}B_l = 0$ and $r_{l,m}A_l^{m-1}B_l=1$. 
	
	Considering the definition of $V_l$ in~\eqref{eq:V_l}, we show by contradiction that $\{-r_{l,m}A_l^j\}_{j\in \{0,\dots,m-1\}}$ are linearly independent. First, assume that they are linearly dependent.
	Then, there exists a nonzero vector $\beta = [\beta_0, \dots, \beta_{m-1} ]^\top\neq \0_{m}$ such that
	\begin{equation}\label{eq:betaVl}
	-\beta^\top V_l=\sum_{j=0}^{m-1}\beta_jr_{l,m}A_l^{j}=\0_{1\times m}.
	\end{equation}
	Multiplying \eqref{eq:betaVl} from the right by $B_l$ yields that
	\begin{equation*}
		\begin{split}
			\sum_{j=0}^{m-2}\beta_j\underbrace{r_{l,m}A_l^{j}B_l}_{=0} + \beta_{m-1}\underbrace{r_{l,m}A_l^{m-1}B_l}_{=1}=0
		\end{split}.
	\end{equation*}
	Thus, $\beta_{m-1}=0$ is implied by \eqref{eq:betaVl}. Then, one can rewrite \eqref{eq:betaVl} by removing the last term:
	\begin{equation}\label{eq:betaVl_min1}
	-\beta^\top V_l=\sum_{j=0}^{m-2}\beta_jr_{l,m}A_l^{j}=\0_{1\times m}.
	\end{equation}
	Again, one can multiply \eqref{eq:betaVl_min1} from the right with $A_lB_l$ to obtain
	\begin{equation*}
		\begin{split}
			\sum_{j=1}^{m-1}\beta_{j-1}r_{l,m}&A_l^{j}B_l=\sum_{j=1}^{m-2}\beta_{j-1}\underbrace{r_{l,m}A_l^{j}B_l}_{=0} \\
			&+ \beta_{m-2}\underbrace{r_{l,m}A_l^{m-1}B_l}_{=1}=\beta_{m-2}=0
		\end{split}.
	\end{equation*}
	By iterating the same procedure, one has that \eqref{eq:betaVl} implies $\beta = \0_m$, which is a contradiction. Therefore, the matrix $V_l$ is invertible.
	\item Given that the matrices $V_l$ are invertible, it is straightforward to see that the matrix $V$ has full row rank, i.e., $\mathrm{rank}(V) = (M-1)m$. Therefore, it is possible to find a right inverse for it. Given the definition of $V^\dagger$ in \eqref{eq:Vdagger}, we can verify that $VV^\dagger = \I_{(M-1)m}$.
	Moreover, from rank-nullity theorem, $\mathrm{dim}(\mathrm{ker}(V)) = Mm-\mathrm{rank}(V) = m$.
	As the full-rank matrix $\Psi\in\RR^{Mm\times m}$ given in \eqref{eq:Psi} satisfies $V\Psi=\0_{(M-1)m\times m}$, we conclude that its columns form a basis for $\mathrm{ker}(V)$.
	
	\item Lemma~\ref{lem:SimultaneousStabilization} shows that the pairs $(A_l,B_l)$ are simultaneously stabilizable if the LP \eqref{eq:LP_suff_prop} is feasible.
	Next, we will prove that the LP in \eqref{eq:LP_suff_prop} is equivalent to the LP in \eqref{eq:LP_suff_prop_red}, which has a smaller number of decision variables and constraints.
	Considering point \ref{itm:V_fullrank}) of this theorem, all solutions $c$ to the equality constraint in \eqref{eq:LP_suff_prop} can be written as \vspace{-0.3cm}
	\begin{equation}\label{eq:c_parametrization}
	c = V^\dagger v + \Psi  w
	\end{equation}
	for a free vector $ w \in \RR^{m}$, i.e., $w$ parametrizes all $c$ solving $Vc=v$.
	On replacing $c$ in the inequality constraint in \eqref{eq:LP_suff_prop} and removing the equality constraint, one obtains the equivalent reduced-order LP in \eqref{eq:LP_suff_prop_red}.
	Furthermore, in view of Lemma~\ref{lem:SimultaneousStabilization}, for a given solution $c$ to \eqref{eq:LP_suff_prop}, $K=c_1^\top V_1+v_1$ is a simultaneously stabilizing control gain.
	From the parameterization of $c$ in \eqref{eq:c_parametrization}, we have $c_1=(V_1^{-1})^\top(v_M^\top-v_1^\top+ w)$.
	Therefore, the common control gain can be calculated as $K=c_1^\top V_1+v_1 = v_M+ w^\top$, concluding the proof.
\end{enumerate}
\end{proof} \vspace{-0.3cm}

 \begin{remark}\label{rem:LP_simple}
 	The LPs~\eqref{eq:LP_suff_prop_red} and~\eqref{eq:LP_suff_prop} are equivalent in spite of the reduction of the number of decision variables from $Mm$ to $m$ and the elimination of equality constraints. As such, the LP in~\eqref{eq:LP_suff_prop_red} can be used to check the simultaneous stabilizability of a larger number of systems compared to~\eqref{eq:LP_suff_prop}.
 \end{remark}

\begin{remark}\label{rem:computation_of_inverses}
  The modified LP formulation \eqref{eq:LP_suff_prop_red}, compared to \eqref{eq:LP_suff_prop}, utilizes matrices $V^\dagger$ and $\Psi$ which, according to~\eqref{eq:Vdagger} and \eqref{eq:Psi}, can be performed by inverting the $m\times m$ matrices $V_l$ associated to individual agents.
\end{remark}

\begin{remark}\label{rem:small_difference_between_subsystems}
  As expected, \eqref{eq:LP_suff_prop_red} is always feasible if 
  $(A_l,B_l)=(A,B)$ for every $l$, as this condition implies
  $v=\0_{(M-1)m}$. Therefore, $h-HV^\dagger v=\1_{M2^m}$ and $ w = \0_m$ is a feasible solution to \eqref{eq:LP_suff_prop_red}. Moreover, $V^\dagger v$ changes continuously with the matrices $A_l$ and $B_l$ if all pairs $(A_l,B_l)$ are controllable. As such, the right-hand side of \eqref{eq:LP_suff_prop_red} is still nonnegative if the differences between the pairs $(A_l,B_l)$ are sufficiently small,
  and $w=\0_m$ is still a feasible solution to \eqref{eq:LP_suff_prop_red}.
  Therefore, a group of controllable systems is always simultaneously stabilizable if the pairs $(A_l,B_l)$ are sufficiently \textit{similar}.
\end{remark}\vspace{-0.4cm}

\section{Conditions for Consensusability of LIMASs}\label{sec:Consensusability}
        
In the sequel, we analyze the consensusability of LIMASs with subsystems of order $n=1$ and $n\geq 1$ separately. Indeed, the former case can be studied without restrictive assumptions while still giving important insight into the consensusability problem. Instead, the latter case is more difficult to analyze and requires additional assumptions. Table~\ref{tab:ConsensusabilityResults} summarizes the results presented in this section along with their applicability conditions. After presenting our results for these two cases in Sections~\ref{subsec:ScalarDynamics} and \ref{subsec:VectorDynamics}, respectively, we discuss the implications of the results in Section~\ref{subsec:DiscussionofResults}.

In the sequel, we use the definitions $\lambda_{p,\max}\triangleq\max_{i\in\{2,\dots,N\}}\lambda_{p,i}$, $\lambda_{p,\min}\triangleq\min_{i\in\{2,\dots,N\}}\lambda_{p,i}$, $\lambda_{c,\max}\triangleq\max_{i\in\{2,\dots,N\}}\lambda_{c,i}$, $\lambda_{c,\min}\triangleq\min_{i\in\{2,\dots,N\}}\lambda_{c,i}$, $\gamma_c \triangleq \frac{\lambda_{c,\max}}{\lambda_{c,\min}}$, and $\Delta_p \triangleq \lambda_{p,\max}-\lambda_{p,\min}$.

We call the scalar $\gamma_c\geq 1$ the \textit{eigenratio} of $\Lcal_c$, where $\gamma_c=1$ if and only if the graph is complete~\cite{you2011}. A low eigenratio means the graph is close to a complete graph. Furthermore, the eigenratio can be decreased by adding edges to the graph, meaning that a $\gamma_c$ close to $1$ generally implies a \textit{densely-connected} graph~\cite{barahona2002synchronization}. $\Delta_p$ denotes the difference between the largest and second smallest eigenvalues of $\Lcal_p$. A low $\Delta_p$ value indicates that the eigenvalues of $\Lcal_p$ are close to each other, which holds when the eigenratio of the physical interconnection graph is low, i.e., $\Gcal_p$ is densely-connected~\cite{you2011, barahona2002synchronization}. A low $\Delta_p$ value is also achieved if eigenvalues of $\Lcal_p$ are small, and consequently, the physical coupling between subsystems is weak. Our sufficient and necessary consensusability conditions are given in terms of these quantities.

\begin{table}
	\caption{Summary of Consensusability Results}
	\label{tab:ConsensusabilityResults}
	\setlength{\tabcolsep}{3pt}
	\begin{tabular}{|p{45pt}|p{45pt}|p{60pt}|p{70pt}|}
		\hline
		\textbf{Result} & \textbf{Type} & \textbf{Feature} & \textbf{Assumptions} \\
		\hline
		Theorem~\ref{prop:consensusability_scalar} & Sufficient & Algebraic Test & Scalar Subsystems \\
		Corollary~\ref{crl:consensusability_numerical} & Sufficient & Linear Program & Assumptions~\ref{ass:lapl_commute},~\ref{ass:controllability} \\
		Theorem~\ref{prop:consensusability_sufficient} & Sufficient & Algebraic Test & Assumptions~\ref{ass:lapl_commute},~\ref{ass:controllability},~\ref{ass:A_and_Ap} \\
		Theorem~\ref{prop:consensusability_necessary} & Necessary & Algebraic Test & Assumptions~\ref{ass:lapl_commute},~\ref{ass:controllability} \\
		\hline
	\end{tabular}\vspace{-0.4cm}
\end{table}\vspace{-0.3cm}

\subsection{Scalar subsystems} \label{subsec:ScalarDynamics}
In this subsection, we present results for scalar subsystems.
Without loss of generality, we assume $B=1$ and denote matrices $A$ and $K$  as $a$ and $k$, respectively.
To simplify the analysis, the matrix $A_p$ is omitted, as it can be lumped into the weights $a_{ij}$.
Then \eqref{eq:reduced-closed-loop-matrix} simplifies to
\begin{equation}\label{eq:deltatilde_dyn_scalar}
a\I_{N-1}- \tilde{\Lcal}_p+k \tilde{\Lcal}_c.
\end{equation}
Next, we present analytical sufficient conditions, which are based on results on the eigenvalues of the sum of two Hermitian matrices.

\begin{theorem}\label{prop:consensusability_scalar}[Scalar System Consensusability Condition]
The LIMAS~\eqref{eq:Si_dyn} with scalar subsystems is consensusable if either of the following conditions holds
        \begin{enumerate}
                \item[S1.] $\lambda_{p,\min}>a-1$ and $\left(\gamma_c-1\right)\left(1-a+\lambda_{p,\min}\right)<\gamma_c\left(2-\Delta_p\right)$,
                \item[S2.] $\lambda_{p,\max}<1+a$ and $\left(\gamma_c-1\right)\left(1+a-\lambda_{p,\max}\right)<\gamma_c\left(2-\Delta_p\right)$.
                \end{enumerate}

Moreover, the control gain $k$ can be selected as $k\in \KK^+\cap\RR_{\geq0}$ if S1 is satisfied and $k\in \KK^-\cap\RR_{<0}$ if S2 is satisfied, where
        \begin{equation}\label{eq:K+K-}
        \begin{split}
        \KK^+ &= \left(\frac{-1-a+\lambda_{p,\max}}{\lambda_{c,\min}},\frac{1-a+\lambda_{p,\min}}{\lambda_{c,\max}}\right) \\
        \KK^- &=  \left(\frac{-1-a+\lambda_{p,\max}}{\lambda_{c,\max}},\frac{1-a+\lambda_{p,\min}}{\lambda_{c,\min}}\right)
        \end{split}.
        \end{equation}
\end{theorem}

\begin{proof}
Since $a\I_{N-1}-\tilde{\Lcal}_p$ and $k\tilde{\Lcal}_c$ are symmetric matrices, upper and lower bounds on the eigenvalues of their summation can be found as in~Theorems~4.3.7 and 4.3.27 in~\cite{horn1985}:
\small
\begin{equation}\label{eq:scalar_eigenvaluebounds}
\begin{split}
\lambda_{\max}\left(a\I_{N-1}-\tilde{\Lcal}_p+k\tilde{\Lcal}_c\right) &\leq \lambda_{\max}\left(a\I_{N-1}-\tilde{\Lcal}_p\right) + \lambda_{\max}\left(k\tilde{\Lcal}_c\right) \\
\lambda_{\min}\left(a\I_{N-1}-\tilde{\Lcal}_p+k\tilde{\Lcal}_c\right) &\geq \lambda_{\min}\left(a\I_{N-1}-\tilde{\Lcal}_p\right) + \lambda_{\min}\left(k\tilde{\Lcal}_c\right) \\
\end{split}.
\end{equation}
\normalsize

Therefore, \eqref{eq:deltatilde_dyn_scalar} is Schur stable if $k \in \RR$ satisfies
\begin{equation}
\label{eq:ktilde_conditions}
\begin{split}
&a-\lambda_{p,\min}+\lambda_{\max}\left(k\tilde{\Lcal}_c\right)<1  \quad \text{and}\\
& a-\lambda_{p,\max}+\lambda_{\min}\left(k\tilde{\Lcal}_c\right)>-1.
\end{split}
\end{equation}

Since the sign of $k$ changes the expression of $\lambda_{\min}(k\tilde{\Lcal}_c)$ and $\lambda_{\max}(k\tilde{\Lcal}_c)$, we inspect the two possibilities $k\geq 0$ and $k<0$ separately.
For $k\ge 0$, we have $\lambda_{\min}(k\tilde{\Lcal}_c) = k\lambda_{c,\min}$ and $\lambda_{\max}(k\tilde{\Lcal}_c) = k\lambda_{c,\max}$, and the conditions~\eqref{eq:ktilde_conditions} simplify to $k\in\KK^+\cap\RR_{\geq0}$, where $\KK^+$ is as given in \eqref{eq:K+K-}.
This is possible if $\KK^+\neq \emptyset$ and $\KK^+\cap\RR_{\geq0}\neq\emptyset$, which directly translate into the conditions in S1.
The result for $k<0$ can be proved similarly, completing the proof.
\end{proof}

\begin{remark}\label{rem:scalar_conservativity}
The only sources of conservativity in Theorem~\ref{prop:consensusability_scalar} are the upper and lower bounds used in~\eqref{eq:scalar_eigenvaluebounds}. Note that the equalities in~\eqref{eq:scalar_eigenvaluebounds} hold when $\tilde{\Lcal}_p=\0_{(N-1)\times (N-1)}$, i.e., there is no physical coupling. This, in turn, means that the conditions S1 and S2 are necessary and sufficient when there is no physical coupling. In this case, Theorem~\ref{prop:consensusability_scalar} recovers the necessary and sufficient condition for consensusability of MASs in~\cite{you2011}.
\end{remark}
As discussed later in Section~\ref{subsec:DiscussionofResults}, conditions S1 and S2 help understanding the roles of the physical interconnection and communication graphs in consensusability. We next analyze the consensusability problem for general subsystems.\vspace{-0.4cm}

\subsection{General subsystems}\label{subsec:VectorDynamics}
From~\eqref{eq:reduced-closed-loop-matrix}, the consensusability problem can be seen as the problem of designing a control gain $\tilde{K}=(\I_{N-1}\otimes K) \in \RR^{(N-1) \times (N-1)n}$ with structural constraints to make the matrix
        $$\I_{N-1}\otimes A-\tilde{\Lcal}_p\otimes A_p+\left(\tilde{\Lcal}_c\otimes B\right)\tilde{K}$$
        Schur stable. Prior work shows that this problem is difficult to tackle without focusing on special system structures~\cite{rotkowitz2005}.
        Therefore, to facilitate the analysis, we introduce the following technical assumption.
\begin{assumption}\label{ass:lapl_commute}
        The Laplacian matrices $\Lcal_p$ and $\Lcal_c$ commute.
\end{assumption}

\begin{remark}\label{rem:commutingLaplacians}
 Assumption~\ref{ass:lapl_commute} is fulfilled when the two Laplacians are equal to each other up to scaling with a constant, i.e., $\Lcal_p=\beta \Lcal_c$, $\beta\in\RR_{\geq 0}$.
  Moreover, two Laplacians commute also when one of them is the Laplacian of a complete graph with uniform edge weights.
  Nevertheless, a necessary and sufficient condition for two generic Laplacians to commute is not yet available in the literature.
\end{remark}

Under Assumption~\ref{ass:lapl_commute}, one can simultaneously diagonalize the two Laplacians $\Lcal_p$ and $\Lcal_c$~\cite{horn1985}, i.e., a unitary transformation matrix $\Phi$ can be chosen such that $\Phi^{\top}\Lcal_p\Phi = \Lambda_p = \mathrm{diag}\left(0,\lambda_{p,2},\dots,\lambda_{p,N}\right)$ and $\Phi^{\top}\Lcal_c\Phi = \Lambda_c = \mathrm{diag}\left(0,\lambda_{c,2},\dots,\lambda_{c,N}\right)$. Note that, for each $i=2,\dots,N$, $\lambda_{p,i}$ and $\lambda_{c,i}$ have the same eigenspace. As such, the diagonal entries of $\Lambda_p$ and $\Lambda_c$, hence $\lambda_{p,i}$ and $\lambda_{c,i}$, are not necessarily ordered by their magnitude.
Assumption~\ref{ass:lapl_commute} thus allows one to decouple the dynamics of $\tilde{\delta}_i$, $i\in \{2,\dots,N\}$ from each other:
\begin{equation}\label{eq:deltatildei_dyn}
\tilde{\delta}_i^+ = \left(A-\lambda_{p,i}A_p+\lambda_{c,i}BK\right)\tilde{\delta}_i \quad \forall i\in\{2,\dots,N\}.
\end{equation}
Consequently, the consensusability of \eqref{eq:LIMAS_dyn} is equivalent to the simultaneous stabilizability of \eqref{eq:deltatildei_dyn}.
The following assumption is required before further derivations.

\begin{assumption}\label{ass:controllability}
        The pairs $(A-\lambda_{p,i}A_p,B)$ are controllable for all $i\in \{2,\dots,N\}$.
\end{assumption}

\begin{remark}
Assumption~\ref{ass:controllability} is necessary as the controllability of the pair $(A-\lambda_{p,i}A_p,B)$ is not implied by that of the pair $(A,B)$ in general.
We also stress that the controllability of $(A-\lambda_{p,i}A_p,B)$ implies that of $(A-\lambda_{p,i}A_p,\lambda_{c,i}B)$, as $\lambda_{c,i}>0,$ $\forall i\in\{2,\dots,N\}$.
\end{remark}

Defining $A_i\triangleq A-\lambda_{p,i}A_p$ and $B_i\triangleq \lambda_{c,i}B$, $i\in\{2,\dots,N\}$, consensusability of LIMAS~\eqref{eq:LIMAS_dyn} is equivalent to the simultaneous stabilizability of pairs $(A_i,B_i)$, which is the problem addressed in Section~\ref{sec:SimultaneousStabilization}.
As mentioned in that section, this problem is difficult to solve in general and we separate our discussion in two parts.
Firstly, a numerical sufficient condition for consensusability is proposed based on Theorem~\ref{thm:SimultaneousStabilization}.
Secondly, for special LIMASs with $A_p=\alpha A$ for some scalar $\alpha$, we give an algebraic sufficient condition for consensusability in Theorem~\ref{prop:consensusability_sufficient}.
Finally, in Theorem~\ref{prop:consensusability_necessary}, we present a necessary condition for the consensusability of \eqref{eq:LIMAS_dyn}.

\subsubsection{Sufficient Conditions}
The following corollary presents an LP-based test for consensusability with generic system matrices $A$, $A_p$, and $B$.
The result directly follows from Theorem~\ref{thm:SimultaneousStabilization}; therefore, the proof is omitted.

\begin{corollary} \label{crl:consensusability_numerical}[Numerical Sufficient Consensusability Condition]
        Suppose that Assumptions~\ref{ass:lapl_commute} and \ref{ass:controllability} hold. The LIMAS \eqref{eq:LIMAS_dyn} is consensusable if there exists a vector $ w \in \RR^n$ satisfying \eqref{eq:LP_suff_prop_red}, where the matrices $V$, $v$, $H$, $h$, $V^\dagger$, and $\Psi$ are computed as in \eqref{eq:LP_suff_prop}-\eqref{eq:V_H_definitions},\eqref{eq:Vdagger}-\eqref{eq:Psi}, by replacing $(A_l,B_l)$ with $(A-\lambda_{p,i}A_p, \lambda_{c,i}B)$. Moreover, if such a $ w$ exists, a control gain achieving consensus is given by $K=v_N+ w^\top$.
\end{corollary}

Next, we provide an analytical sufficient condition for consensusability of LIMASs. As discussed in Section~\ref{sec:SimultaneousStabilization}, it is difficult to provide analytical conditions on simultaneous stabilizability of more than three systems~\cite{blondel1993}. For this reason, we limit our analysis to LIMASs verifying the following condition.

\begin{assumption}\label{ass:A_and_Ap}
        The physical interconnection matrix $A_p$ satisfies $A_p = \alpha A$, where $\alpha\in \RR$.
\end{assumption}

Under this assumption, consensusability is equivalent to the simultaneous stabilizability of pairs $((1-\alpha\lambda_{p,i})A,\lambda_{c,i}B)$, i.e., $A_i$ and $B_i$ are now only characterized by scalar multiplications of common matrices $A$ and $B$, respectively. This allows for easier analysis of Schur stability of matrices $A_i+B_iK$, for a given control gain $K$. For this purpose, we first define $\alpha_i\triangleq 1-\alpha\lambda_{p,i}$, $\alpha_{\max}\triangleq \max_i|\alpha_i|$, $\alpha_{\min}\triangleq \min_i|\alpha_i|$, $\bar{A}\triangleq \alpha_{\max}A$, and
        \begin{equation}\label{eq:sigma_c}
        \sigma_c\triangleq
        1-\frac{1}{\prod_{k}\left|\lambda_k^u\left(\bar{A}\right)\right|^2},
        \end{equation}
        where $\lambda^u_1(\bar{A})$, $\lambda^u_2(\bar{A})$, $\dots$ denote the unstable eigenvalues of $\bar{A}$.
Moreover, in the sequel, we leverage the results in~\cite{schenato2007} stating that if $\sigma>\sigma_c$, there exists a matrix $P=P^\top\succ 0$ solving the modified algebraic Riccati equation (MARE)
        \begin{equation}
		\label{eq:MARE}
        \begin{split}
                \bar{A}^{\top}P\bar{A}-\sigma\bar{A}^{\top}PB\left(B^{\top}PB\right)^{-1}B^{\top}P\bar{A}-P\prec 0
        \end{split}.
        \end{equation}

        The following theorem presents an algebraic sufficient condition for consensusability under Assumption~\ref{ass:A_and_Ap}. Note that, unlike Corollary~\ref{crl:consensusability_numerical}, following results do not require the knowledge of the eigenvalue pairs $(\lambda_{p,i},\lambda_{c,i})$ associated to the same eigenspace.

\begin{theorem}\label{prop:consensusability_sufficient}[Analytical Sufficient Consensusability Condition]
  Suppose that Assumptions~\ref{ass:lapl_commute}, \ref{ass:controllability}, and \ref{ass:A_and_Ap} hold.
  If  $\bar{A}$ is Schur stable, the LIMAS in \eqref{eq:LIMAS_dyn} is consensusable by using the control gain $K=\0_{1\times n}$.
  Besides, if $\bar{A}$ is not Schur stable, and the following condition holds
        \begin{equation}\label{eq:consensusability_sufficient_critical}
        \left(\frac{\max_{i,j}\frac{\alpha_i}{\lambda_{c,j}}-\min_{i,j}\frac{\alpha_i}{\lambda_{c,j}}}{2}\right)^2<\frac{\alpha_{\min}^2-\alpha_{\max}^2\sigma_c}{\lambda_{c,\max}^2},
        \end{equation}
        the LIMAS in \eqref{eq:LIMAS_dyn} is consensusable by using the control gain $K=-k^*(B^{\top}PB)^{-1}B^{\top}PA$ where $k^*=\frac{\min_{i,j}\frac{\alpha_i}{\lambda_{c,j}}+\max_{i,j}\frac{\alpha_i}{\lambda_{c,j}}}{2}$ and $P$ is the solution to the MARE \eqref{eq:MARE} with $\sigma=                \min_{i,j}\frac{2\alpha_i\lambda_{c,j}k^*-\lambda_{c,j}^2(k^*)^2}{\alpha_{\max}^2}$.
\end{theorem}

\begin{proof}
Under Assumptions~\ref{ass:lapl_commute} and \ref{ass:A_and_Ap}, consensusability of \eqref{eq:LIMAS_dyn} is equivalent to the simultaneous stabilizability of
\begin{equation}\label{eq:deltatilde_dyn_assumptionAp}
\tilde{\delta}_i^+ = \left(\alpha_iA+\lambda_{c,i}BK\right)\tilde{\delta}_i \quad \forall i\in\{2,\dots,N\}.
\end{equation}
It is straightforward to see that, if $\bar{A}$ is Schur stable, so are $\alpha_iA$; therefore, the LIMAS \eqref{eq:LIMAS_dyn} is consensusable with the control gain $K=\0_{1\times n}$.

Next, we will show that when $\bar{A}$ is not Schur stable, and if \eqref{eq:consensusability_sufficient_critical} holds, the LIMAS \eqref{eq:LIMAS_dyn} is consensusable by the control gain designed in the theorem. For this purpose, we will first show in the sequel that if \eqref{eq:consensusability_sufficient_critical} holds, we have
\begin{equation}\label{eq:min_ij}
\min_{i,j}\frac{2\alpha_i\lambda_{c,j}k^*-\lambda_{c,j}^2(k^*)^2}{\alpha_{\max}^2}>\sigma_c.
\end{equation}

Defining a function $f(k) \triangleq \max_{i,j}|k-\frac{\alpha_i}{\lambda_{c,j}}| \geq 0$, it is straightforward to calculate its optimizer
$$k^*=\arg \min_kf(k)=\frac{\min_{i,j}\frac{\alpha_i}{\lambda_{c,j}}+\max_{i,j}\frac{\alpha_i}{\lambda_{c,j}}}{2}.$$
Moreover, one can show that\vspace{-0.2cm}
\begin{equation*}
	\arg \min_k \max_{i,j}\left(k-\frac{\alpha_i}{\lambda_{c,j}}\right)^2=\arg \min_{k}f(k)=k^*;
\end{equation*}
therefore, it directly follows that\vspace{-0.2cm}
\begin{equation*}
	\begin{split}
		\min_k \max_{i,j}\left(k-\frac{\alpha_i}{\lambda_{c,j}}\right)^2
		&=\left(\frac{\max_{i,j}\frac{\alpha_i}{\lambda_{c,j}}-\min_{i,j}\frac{\alpha_i}{\lambda_{c,j}}}{2}\right)^2
	\end{split}.
\end{equation*}
Thus, in view of \eqref{eq:consensusability_sufficient_critical}, we have\vspace{-0.2cm}
\begin{equation}\label{eq:minmax}
\max_{i,j}\left(k^*-\frac{\alpha_i}{\lambda_{c,j}}\right)^2 <\frac{\alpha_{\min}^2-\alpha_{\max}^2\sigma_c}{\lambda_{c,\max}^2}.
\end{equation}
Noticing that\vspace{-0.2cm}
$$\lambda_{c,j}^2\left(k^*-\frac{\alpha_i}{\lambda_{c,j}}\right)^2 = -(2\alpha_i\lambda_{c,j}k^*-\lambda_{c,j}^2(k^*)^2) +\alpha_i^2,$$
we get \vspace{-0.2cm}
\small
\begin{equation}\label{eq:minusminleq}
\begin{split}
\max_{i,j}\lambda_{c,j}^2\left(k^*-\frac{\alpha_i}{\lambda_{c,j}}\right)^2 \geq -\min_{i,j}\left(2\alpha_i\lambda_{c,j}k^*-\lambda_{c,j}^2(k^*)^2\right)+\alpha_{\min}^2.
\end{split}
\end{equation} \vspace{-0.5cm}
\normalsize

\noindent Furthermore, from the fact that $$ \lambda_{c,\max}^2\max_{i,j}\left(k^*-\frac{\alpha_i}{\lambda_{c,j}}\right)^2\ge \max_{i,j}\lambda_{c,j}^2\left(k^*-\frac{\alpha_i}{\lambda_{c,j}}\right)^2,$$
and using \eqref{eq:minmax}, \eqref{eq:minusminleq}, we can obtain
\begin{equation}\label{eq:alphamax_sigmac}
-\min_{i,j}\left(2\alpha_i\lambda_{c,j}k^*-\lambda_{c,j}^2(k^*)^2\right)+\alpha_{\min}^2 < \alpha_{\min}^2-\alpha_{\max}^2\sigma_c.
\end{equation}
This inequality is equivalent to \eqref{eq:min_ij}; therefore, a solution $P$ to the MARE~\eqref{eq:MARE} exists with $$\sigma= \min_{i,j}\frac{2\alpha_i\lambda_{c,j}k^*-\lambda_{c,j}^2(k^*)^2}{\alpha_{\max}^2}>\sigma_c.$$
Further, defining $\sigma_i\triangleq \frac{2\alpha_i\lambda_{c,i}k^*-\lambda_{c,i}^2(k^*)^2}{\alpha_{\max}^2}$\vspace{0.05cm}, we see that $\sigma \leq \sigma_i$ holds by definition. Therefore, $P$ also satisfies the following for each $i$, \vspace{-0.2cm}
\begin{equation*}
	\begin{split}
		\bar{A}^{\top}P\bar{A}-\sigma_i\bar{A}^{\top}PB\left(B^{\top}PB\right)^{-1}B^{\top}P\bar{A}-P\prec 0
	\end{split},
\end{equation*}\vspace{-0.6cm}

\noindent which further implies
\small
\begin{equation}\label{eq:lyapunovcond_max}
\begin{split}
\alpha_{\max}^2A^{\top}PA+&\left(\lambda_{c,i}^2(k^*)^2-2\alpha_i\lambda_{c,i}k^*\right) \\
&A^{\top}PB\left(B^{\top}PB\right)^{-1}B^{\top}PA 
-P\prec 0.
\end{split}
\end{equation}
\normalsize
This is equivalent to the existence of $P\succ 0$ and $K = -k^*(B^{\top}PB)^{-1}B^{\top}PA$ such that $$\left(\alpha_iA+\lambda_{c,i}BK\right)^{\top}P\left(\alpha_iA+\lambda_{c,i}BK\right)-P\prec 0,$$ for all $i\in\{2,\dots,N\}$.
Therefore, the systems~\eqref{eq:deltatilde_dyn_assumptionAp} are simultaneously stabilizable, which further shows that the LIMAS \eqref{eq:LIMAS_dyn} is consensusable by the designed controller.
\end{proof}

The algebraic conditions in Theorem~\ref{prop:consensusability_sufficient} are easily verifiable and do not require to solve any optimization problem. Moreover, they provide important insights about the effects of physical interconnection and communication graphs on consensusability, as discussed in Section~\ref{subsec:DiscussionofResults}. We stress that, this result is conservative because the derivations \eqref{eq:min_ij}-\eqref{eq:alphamax_sigmac} consider the worst-case scenario in terms of eigenvalue pairs $(\lambda_{p,i},\lambda_{c,i})$. Indeed, without physical couplings ($\alpha=0$), the inequality~\eqref{eq:consensusability_sufficient_critical} cannot recover the sufficient and necessary condition for consensusability of MASs proposed in~\cite{you2011}.

\subsubsection{Necessary Conditions}
We now provide algebraic necessary conditions for the consensusability of the LIMAS in~\eqref{eq:LIMAS_dyn}, which will allow us to identify features of the LIMAS disfavoring consensusability.
We note that the following theorem does not require Assumption~\ref{ass:A_and_Ap}.

\begin{theorem}\label{prop:consensusability_necessary}[Necessary Consensusability Condition]
        Suppose that Assumptions~\ref{ass:lapl_commute} and \ref{ass:controllability} hold. If the LIMAS in \eqref{eq:LIMAS_dyn} is consensusable, at least one of the following conditions N1-N3 holds:
        \begin{enumerate}
                \item [N1.] $\max_i|\det(A_i)|<1$,
                \item [N2.] $\max_i|\det(A_i)|\geq1$ and $\min_i|\det(A_i)|<1$ and \eqref{eq:N2_secondeq},
	            \item [N3.] $\min_i|\det(A_i)|\geq1$ and \eqref{eq:N2_secondeq} and \eqref{eq:N3_thirdeq},
        \end{enumerate}
    	where $A_i = A-\lambda_{p,i}A_p$ and 
    	\begin{equation}\label{eq:N2_secondeq}
    	\max_i|\det(A_i)|-1 < \gamma_c\min_i|\det(A_i)|+\gamma_c,
    	\end{equation}
    	\begin{equation}\label{eq:N3_thirdeq}
    	\gamma_c(\min_i|\det(A_i)|-1) < \max_i|\det(A_i)|+1.
    	\end{equation}
\end{theorem}

\begin{proof}
The proof is a modification of the proof of Lemma 3.1 in \cite{you2011}.
Under Assumption~\ref{ass:controllability}, without loss of generality, each pair $(A-\lambda_{p,i}A_p,B)$ can be written in controllable canonical form \vspace{-0.2cm}
\begin{equation*}
	A-\lambda_{p,i}A_p = \begin{bmatrix}
		0 & 1 & 0 & \dots  \\
		\vdots & \ddots & \ddots & \ddots \\
		0 & \dots & 0 & 1 \\
		-a_{i,1} & -a_{i,2} & \dots & -a_{i,n}
	\end{bmatrix}, \enskip B = \begin{bmatrix}
		0 \\
		\vdots \\
		0 \\
		1
	\end{bmatrix},
\end{equation*}
where $|a_{i,1}| = |\det(A-\lambda_{p,i}A_p)|$.
Given a simultaneously stabilizing feedback gain $K = \left[k_1, \dots, k_N\right]$, one can see that $|\det(A_{cl,i})| = |a_{i,1}-\lambda_{c,i}k_1|$ for the closed-loop matrix $A_{cl,i}=A-\lambda_{p,i}A_p+\lambda_{c,i}BK$.
Since $K$ is selected to stabilize $(A-\lambda_{p,i}A_p,\lambda_{c,i}B)$ for all $i \in \{2,\dots,N\}$, it holds that $\rho(A_{cl,i})<1$. Therefore, it holds that $|\det(A_{cl,i})|=\prod_{j}\lambda_j(A_{cl,i})<1$,
yielding $|a_{i,1}-\lambda_{c,i}k_1|<1$, which can be further manipulated to give \vspace{-0.15cm}
\begin{equation*}
	\frac{|a_{i,1}|-1}{\lambda_{c,i}}<|k_1|< \frac{|a_{i,1}|+1}{\lambda_{c,i}}.
\end{equation*}\vspace{-0.2cm}

\noindent This implies that $\bigcap_{i=2}^N\left(\frac{|a_{i,1}|-1}{\lambda_{c,i}},\frac{|a_{i,1}|+1}{\lambda_{c,i}}\right)\neq \emptyset$; therefore, 
\begin{equation}\label{eq:a_i_bounds}
\max_i\frac{|a_{i,1}|-1}{\lambda_{c,i}} < \min_i\frac{|a_{i,1}|+1}{\lambda_{c,i}}.
\end{equation}
Below, we will show that the above inequality implies that at least one of the conditions N1-N3 holds.
We start by noting that, the left-hand side of \eqref{eq:a_i_bounds} can either be negative or nonnegative. In the former case, \eqref{eq:a_i_bounds} is always satisfied and it holds that $\max_i|a_{i,1}|<1$, 
yielding the condition N1. On the contrary, when the left-hand side of \eqref{eq:a_i_bounds} is nonnegative,  $$\max_i\frac{|a_{i,1}|-1}{\lambda_{c,i}}\geq\max_i\frac{|a_{i,1}|-1}{\max_j\lambda_{c,j}}=\frac{\max_i|a_{i,1}|-1}{\lambda_{c,\max}}$$ and \vspace{-0.2cm} $$\min_i\frac{|a_{i,1}|+1}{\lambda_{c,i}}\leq\min_i\frac{|a_{i,1}|+1}{\min_j\lambda_{c,j}}=\frac{\min_i|a_{i,1}|+1}{\lambda_{c,\min}}$$hold. Therefore, combining these inequalities with \eqref{eq:a_i_bounds}, one gets \eqref{eq:N2_secondeq}.
In addition, we note that 
\begin{equation*}
	\max_i\frac{|a_{i,1}|-1}{\lambda_{c,i}} \geq \max_i\frac{\min_j|a_{j,1}|-1}{\lambda_{c,i}},
\end{equation*}
which means that \eqref{eq:a_i_bounds} implies 
\begin{equation}\label{eq:a_i_bounds_impl}
\max_i\frac{\min_j|a_{j,1}|-1}{\lambda_{c,i}} < \min_i\frac{|a_{i,1}|+1}{\lambda_{c,i}}.
\end{equation}
We again make the distinction of two cases where $\min_j|a_{j,1}|-1$ is negative or nonnegative. When the former holds, \eqref{eq:a_i_bounds_impl} is always satisfied. Hence, combining $\min_j|a_{j,1}|<1$ with $\max_i|a_{i,1}|\geq1$ and \eqref{eq:N2_secondeq}, one forms the condition N2. Otherwise, when it is nonnegative, one can show that
\begin{equation*}
	\max_i\frac{\min_j|a_{j,1}|-1}{\lambda_{c,i}} = \frac{\min_j|a_{j,1}|-1}{\lambda_{c,\min}}.
\end{equation*}
We can also derive an upper bound to the term on the right-hand side of \eqref{eq:a_i_bounds_impl} as
$$\min_i\frac{|a_{i,1}|+1}{\lambda_{c,i}}\leq\min_i\frac{\max_j|a_{j,1}|+1}{\lambda_{c,i}}=\frac{\max_j|a_{j,1}|+1}{\lambda_{c,\max}}.$$ Incorporating the last two equations with \eqref{eq:a_i_bounds_impl}, we get \eqref{eq:N3_thirdeq}.
Finally, condition N3 is obtained by combining \eqref{eq:N2_secondeq} and \eqref{eq:N3_thirdeq} with $\min_j|a_{j,1}|\geq1$.
\end{proof}

Note that the conditions N1-N3 involve only the three quantities, $\max_i|\det(A_i)|$, $\min_i|\det(A_i)|$, and $\gamma_c$, making them easy to verify. These conditions also provide an understanding of how these variables effect consensusability, as discussed in the next section.  Finally, observe that N3 reduces to the necessary and sufficient condition for the consensusability of MASs provided in~\cite{you2011} when physical couplings are absent, i.e., when $A_p=\0_{n\times n}$. 

We also highlight that, in addition to providing sufficient conditions for consensusability, Theorems~\ref{prop:consensusability_scalar} and \ref{prop:consensusability_sufficient} as well as Corollary~\ref{crl:consensusability_numerical} also show how to design the controller gain $K$ for reaching consensus. \vspace{-0.6cm}

\subsection{Discussion of the consensusability results}\label{subsec:DiscussionofResults}

Hereafter, we discuss implications of the results given in the previous section. Specific comments provided below point out that consensusability is easier to achieve in LIMASs with
\begin{enumerate}
	\item \textit{weak} physical coupling and
	\item \textit{densely-connected} physical and communication graphs.
\end{enumerate}

Specifically, the conditions S1 and S2 in Theorem~\ref{prop:consensusability_scalar} 
can be satisfied only if $\Delta_p < 2$. This is possible if $\lambda_{p,\max}$ is close to $\lambda_{p,\min}$, or if the physical coupling is weak. We also note that, by assuming $\Delta_p<2$ and that the first inequality of S1 is satisfied, $\gamma_c$ values satisfying the second inequality can always be found in a neighborhood of $\gamma_c=1$. Moreover, this also holds for condition S2, showing that communication graphs whose Laplacian have eigenvalues that are close to each other favor consensusability. These conditions imply that both graphs $\Gcal_p$ and $\Gcal_c$ are \textit{densely connected}.
Therefore, this feature favors consensusability.

The LP in \eqref{eq:LP_suff_prop_red} is \textit{easier} to solve when the pairs $(A_i,B_i)$ are \textit{closer} to each other, as discussed in Remark~\ref{rem:small_difference_between_subsystems}. With the definition $A_i=A-\lambda_{p,i}A_p$, it is straightforward to see that these matrices are \textit{close} to each other when $\|\lambda_{p,i}A_p\|$ is \textit{small}, or $\lambda_{p,i}$ are similar. Analogously, $B_i=\lambda_{c,i}B$ are \textit{close} to each other when $\lambda_{c,i}$ are similar. These observations point out once more that the LP in \eqref{eq:LP_suff_prop_red} is more likely to be feasible when the physical interconnection is weak and both physical and communication graphs are densely connected.

Considering Theorem~\ref{prop:consensusability_sufficient}, the condition~\eqref{eq:consensusability_sufficient_critical} can be satisfied only if 
$\alpha_{\min}^2>\alpha_{\max}^2\sigma_c$. Taking into account that $\sigma_c$ increases with $\alpha_{\max}$, the inequality is fulfilled only when $\alpha_{\max}$ is sufficiently small and close to $\alpha_{\min}$. Note that $\alpha_{\max}$ takes small values for small values of $\alpha$, i.e., when the effect of physical coupling is weak. Moreover, $\alpha_{\min}$ and $\alpha_{\max}$ are close to each other when $\lambda_{p,\min}$ and $\lambda_{p,\max}$ are close. Also note that the left-hand side of the inequality~\eqref{eq:consensusability_sufficient_critical} decreases as the ratio $\frac{\lambda_{c,\min}}{\lambda_{c,\max}}$ increases, i.e., as the eigenvalues of the communication graph get closer to each other. Therefore, the implications of the Theorem~\ref{prop:consensusability_sufficient} match with the observations made for Theorem~\ref{prop:consensusability_scalar}.

 In order to show the role of physical coupling in Theorem~\ref{prop:consensusability_necessary}, we look at the extreme case $A=\0_{n\times n}$, for which condition N1 is not satisfied for \textit{strong} physical coupling characterized by large values of $|\det(A_p)|$ and $\lambda_{p,i}$. Similarly, the inequality \eqref{eq:N2_secondeq} is more difficult to satisfy for \textit{stronger} physical coupling for fixed $\gamma_c$, since $\max_i|\det(A_i)|$ will be much larger than $\min_i|\det(A_i)|$. One can also see that, for fixed $A_i$, the inequality \eqref{eq:N3_thirdeq} gets more difficult to satisfy as $\gamma_c$ grows, which corresponds to a decrease in the connectivity of the communication graph~\cite{barahona2002synchronization}. This, in turn, means that, it is more difficult to satisfy the condition N3 as the communication graph gets more \textit{sparse}.\vspace{-0.4cm}

\section{Simulation Results}\label{sec:SimulationResults}

 \def\Clone{(-4.0,1.8)}
 \def\Cltwo{(3.0,1.5)}
 \def\Clthree{(-0.85,-1.0)}
 \def\curvrad{0.2}

\begin{figure}[t]
	\centering
	\begin{subfigure}{0.48\textwidth}
		\centering
		\ctikzset{bipoles/length=0.8cm}
		\tikzstyle{every node}=[minimum size=.8cm, inner sep=.8, line width=0.5pt, scale=0.7]
		\begin{circuitikz}[american currents, scale=0.5, line width=1pt]
			\draw \Clone+(0,0) node(Cluster1) [circle, dashed, minimum size=3cm, fill=red!10,draw=black] {};
			\draw \Clone+(0,2.5) node(Title1) {Cluster $1$};
			\draw \Cltwo+(0,0) node(Cluster2) [circle, dashed, minimum size=3.5cm, fill=red!10,draw=black] {};
			\draw \Cltwo+(0,2.8) node(Title2) {Cluster $2$};
			\draw \Clthree+(0,0) node(Cluster3) [circle, dashed, minimum size=2.75cm, fill=red!10,draw=black] {};
			\draw \Clthree+(0,+2.3) node(Title3) {Cluster $3$};
			
			\draw \Clone+(-0.5,0.866) node(s1) [circle, draw=black, double,fill=black!20] {$\Scal_1$};
			\draw \Clone+(1,0) node(s2) [circle, draw=black, double,fill=black!20] {$\Scal_2$};
			\draw \Clone+(-0.5,-0.866) node(s3) [circle, draw=black, double,fill=black!20] {$\Scal_3$};
			
			\draw[blue] (s1) to(s3);
			\draw[blue] (s2) to(s3);
			
			\draw[red,bend left=30] (s1.north east) to (s2.north);
			\draw[red,bend left=30] (s3.west) to (s1.west);
			
			\draw \Cltwo+(-0.95,0.95) node(s4) [circle, draw=black, double,fill=black!20] {$\Scal_4$};
			\draw \Cltwo+(0.95,0.95) node(s5) [circle, draw=black, double,fill=black!20] {$\Scal_5$};
			\draw \Cltwo+(-0.95,-0.95) node(s6) [circle, draw=black, double,fill=black!20] {$\Scal_6$};
			\draw \Cltwo+(0.95,-0.95) node(s7) [circle, draw=black, double,fill=black!20] {$\Scal_7$};
			
			\draw[blue] (s4) to(s7);
			\draw[blue] (s5) to(s6);
			\draw[blue] (s6) to(s7);
			\draw[red,bend left=30] (s2.north east) to (s4.north west);
			
			\draw[red,bend right=30] (s4.south west) to (s6.north west);
			\draw[red,bend right=30] (s6.south east) to (s7.south west);
			\draw[red,bend right=30] (s7.north east) to (s5.south east);
			\draw[red,bend right=30] (s5.north west) to (s4.north east);
			
			\draw \Clthree+(-0.9,0) node(s8) [circle, draw=black, double,fill=black!20] {$\Scal_8$};
			\draw \Clthree+(0.9,0) node(s9) [circle, draw=black, double,fill=black!20] {$\Scal_9$};
			
			\draw[blue] (s8) to(s9);
			
			\draw[red,bend right=30] (s3.south) to (s8.west);
			\draw[red,bend right=30] (s9.east) to (s6.south);
			
			\draw[red,bend left=30] (s8.north east) to (s9.north west);
		\end{circuitikz}
		\caption{Schematic of the network of identical supercapacitors arranged in multiple clusters. Lines in blue and red represent, respectively, the physical and communication interconnections.
		}
		\label{fig:supercapacitornetwork}
	\end{subfigure}

	\begin{subfigure}{0.48\textwidth}
		\centering
		\small
		\tikzstyle{every node}=[scale=0.7]
		\ctikzset{bipoles/length=0.8cm}
		\begin{circuitikz}[american currents, scale=0.5, line width=0.7pt]
			\draw (-3,0.9) node(s8) [circle, double, minimum size=3.4cm, fill=gray!10,draw=black!30] {};
			\draw (-3.0,3.5) node(s8Title) {$\Scal_8$};
			\draw (3,0.9) node(s9) [circle, double, minimum size=3.4cm, fill=gray!10,draw=black!30] {};
			\draw (3,3.5) node(s8Title) {$\Scal_9$};			
			\draw (-2.75,0) [short] to (-4.4,0)
			to [C=$C$, o-o](-4.4,2)
			to [short] (-2.75,2)
			to [R=$R$] (-2.75,0)
			to [short] (-1.15,0)
			to [dcisource, l=$I_8$] (-1.15,2);
			\draw[-latex] (-3.75,0.25) -- (-3.75,1.75)node[midway,right]{$V_{8}$};
			\draw (-2.75,2) [short] to (-1.15,2);
			
			\draw (2.75,0) [short] to (4.4,0)
			to [C, a=$C$, o-o](4.4,2)
			to [short] (2.75,2)
			to [R, a=$R$] (2.75,0)
			to [short] (1.15,0)
			to [dcisource, a=$I_9$] (1.15,2);
			\draw[-latex] (3.75,0.25) -- (3.75,1.75)node[midway,left]{$V_{9}$};
			\draw (2.75,2) [short] to (1.15,2);
			
			\draw (-1.15,2) [R=\textcolor{blue}{$R_{89}$},color=blue, draw=blue]  to (1.15,2);
		\end{circuitikz}
		\caption{Electrical scheme of supercapacitors $\Scal_8$ and $\Scal_9$ in Cluster~$3$.
		}
		\label{fig:supercapacitordynamics}
	\end{subfigure}
	\caption{Network of identical supercapacitors considered for simulations in Section~\ref{subsec:SimulationResults_Scalar}.}\vspace{-0.6cm}
	\label{fig:supercapacitor}
\end{figure}
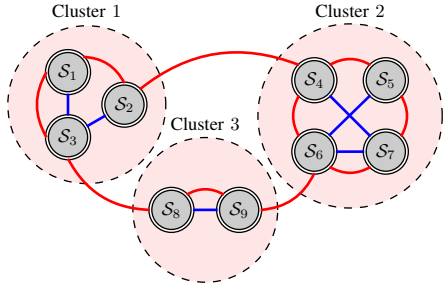
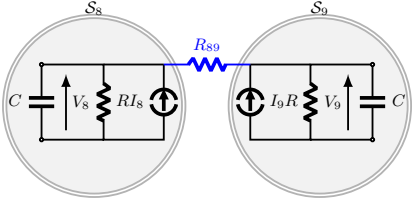

In this section, we present two different sets of simulations to validate the results in Sections~\ref{subsec:ScalarDynamics} and \ref{subsec:VectorDynamics}, respectively.\vspace{-0.3cm}

\subsection{Consensusability of a network of supercapacitors}\label{subsec:SimulationResults_Scalar}

In recent years, supercapacitors have been popularized as an alternative to batteries in application domains such as microgrids (mGs), transportation systems, and automotive~\cite{inthamoussou2013control,freeborn2015fractional}. In some of these applications, a network of multiple supercapacitors can be utilized, making it a LIMAS. In the sequel, we consider a network of identical supercapacitors, which are arranged in clusters, as shown in Figure~\ref{fig:supercapacitornetwork}. Such a system represents the scenario where groups of supercapacitors are far away from each other and no physical interconnection between them is possible.

As shown in Figure~\ref{fig:supercapacitordynamics}, we model each subsystem $i$ as a parallel RC circuit with $C=10F$, $R=5k\Omega$, and a current source supplying a time-varying current $I_i$. The resistance models the power \textit{leak} from the capacitor. Subsystems are coupled through resistive electrical lines (see Figure~\ref{fig:supercapacitordynamics}), whose resistance values $R_{ij}$ are chosen randomly from the interval $[10,50]\Omega$. We are interested in the problem of controlling the charging currents $I_i$ such that the voltages across each supercapacitor reach consensus. In practical applications, this might be needed to ensure the same voltage level across all storage devices, as they might feed the same load during a \textit{discharge} period (not considered here). We assume that, the input currents $I_i$ are computed using the controller~\eqref{eq:Si_contr}, utilizing the communication network represented in Figure~\ref{fig:supercapacitornetwork} with unit edge weights. Applying Kirchoff's current and voltage laws, the voltage dynamics of the supercapacitor $i$ is given by
\small
\begin{equation}\label{eq:supercapacitor_dyn}
        C\dot{V}_i = -\frac{1}{R}V_i - \sum_{j\in\Ncal_i^p}\frac{1}{R_{ij}}\left(V_i-V_j\right) +k\sum_{j\in\Ncal_i^c}\left(V_i-V_j\right).
\end{equation}
\normalsize
In order to match the model~\eqref{eq:LIMAS_dyn}, we discretize~\eqref{eq:supercapacitor_dyn} in time by using the forward Euler method with a sampling period of $T_s = 0.1ms$. We obtain the dynamics\vspace{-0.2cm} $$x^+ = (a\I_N-\Lcal_p^m+k\Lcal_c^m)x,$$\vspace{-0.5cm} 

\noindent where 
 $x=[V_1,\dots,V_9]^\top$, $a=1-T_s\frac{1}{RC}\approx 1$, $\Lcal_p^m\triangleq\frac{T_s}{C}\Lcal_p$, and $\Lcal_c^m\triangleq\frac{T_s}{C}\Lcal_c$. Moreover, $\Lcal_p$ and $\Lcal_c$ are the Laplacian matrices of the physical and communication graphs, respectively. The \textit{modified} Laplacian matrices $\Lcal_p^m$ and $\Lcal_c^m$ account for the effect of discretization, and are used in the analysis of consensusability.

\begin{figure}[t]
        \centering
        \includegraphics[width=0.49\textwidth]{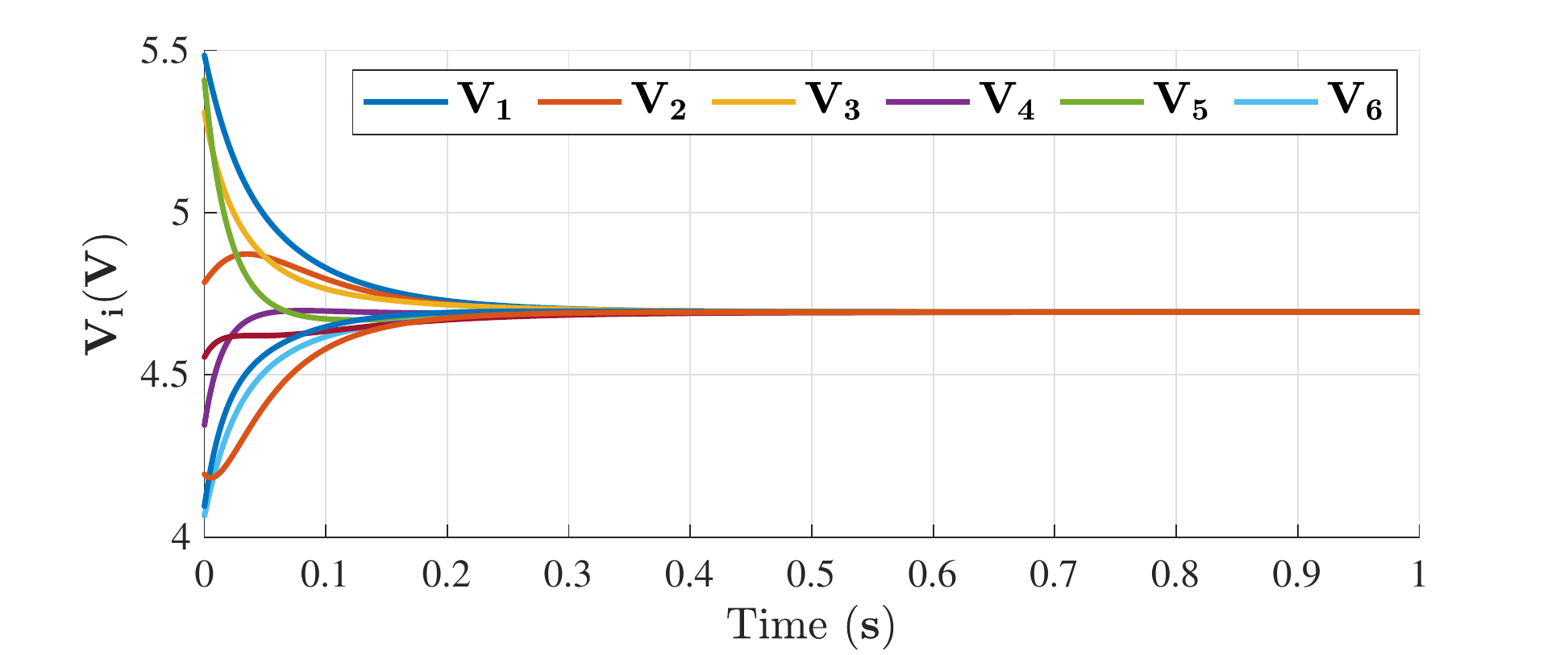}	\vspace{-.6cm}
        \caption{Voltages of the supercapacitors initialized with random initial voltages.} \vspace{-0.6cm}
        \label{fig:supercapacitor_voltages}
 \end{figure} 

 On applying Theorem~\ref{prop:consensusability_scalar} to this system, we see that both conditions S1 and S2 are satisfied; therefore, any a control gain $k$ in the intervals $\mathbb{K}^+\cap\RR_{\geq 0} = [0,4.071\times 10^{-5})$ and $\mathbb{K}^-\cap\RR_{< 0}=(-4.071\times 10^{4},0)$ guarantees the achievement of consensus. We choose $k=-200$ and initialize each supercapacitor with a random voltage value between $4V$ and $6V$ when running the continuous-time simulations of the LIMAS. As shown in Figure~\ref{fig:supercapacitor_voltages}, all voltage levels converge to a common value, as guaranteed by Theorem~\ref{prop:consensusability_scalar}. In order to see how tight the computed interval $\mathbb{K}=(-4.071\times 10^{4},4.071\times 10^{-5})$ of control gains is, we create a fine grid of $k$ values strictly including $\mathbb{K}$ to check the values for which the matrix in~\eqref{eq:deltatilde_dyn_scalar} is Schur stable. We find that consensus is reached for $k$ values in the interval $(-4.071\times 10^{4},1.532\times 10^{-4})$, which is very close to $\mathbb{K}$. This is expected since $\Lcal_p^m$ is close to a zero matrix due to the small sampling time $T_s$~(see Remark~\ref{rem:scalar_conservativity}).
 
 Finally, because this LIMAS is consensusable and there exists a control gain $k$ reaching consensus, Theorem~\ref{prop:consensusability_necessary} asserts that at least one of the conditions N1-N3 should hold. Upon setting $A_i = a-\lambda_i(\Lcal_p^m)$, one sees that N1 is satisfied. \vspace{-0.5cm}

\subsection{Consensusability of a DC microgrid}\label{subsec:SimulationResults_Vector}

DC microgrids (DCmGs) are electrical networks of distributed generation units (DGUs) incorporating renewable energy sources and/or storage devices. They have recently risen in popularity thanks to the advantages they offer over traditional power grid in certain application scenarios~\cite{guerrero2010}. A remarkable feature of DCmGs is their ability to run in islanded mode, i.e., disconnected from the main grid, which brings about the additional challenge of ensuring their safe and reliable operation~\cite{dragivcevic2015dcII,gallo2020distributed}. To this purpose, hiearchical control structures are commonly used, where primary controllers regulate the voltages of DGUs and higher-level controllers enable coordination among DGUs~\cite{nahata2020passivity,tucci2018}. An example is provided by secondary controllers for reaching consensus on DGU states which, in turn, results in current sharing and voltage balancing~\cite{tucci2018}.

\begin{figure}[t]
	\centering
	\begin{subfigure}{0.48\textwidth}
		\centering
		\ctikzset{bipoles/length=0.8cm}
		\tikzstyle{every node}=[minimum size=.8cm, inner sep=.8, line width=0.5pt]
		\begin{circuitikz}[american currents, scale=0.73, line width=1pt]

			\draw (0,0) node(s1) [circle, draw=black, double,fill=black!20] {\tiny $DGU_1$};
			\draw (s1)+(2,0) node(s2) [circle, draw=black, double,fill=black!20] {\tiny $DGU_2$};
			\draw (s2)+(-1,-1.4) node(s3) [circle, draw=black, double,fill=black!20] {\tiny $DGU_3$};
			\draw (s2)+(2,0) node(s4) [circle, draw=black, double,fill=black!20] {\tiny $DGU_4$};
			\draw (s4)+(1,-1.4) node(s5) [circle, draw=black, double,fill=black!20] {\tiny $DGU_5$};
			\draw (s4)+(2,0) node(s6) [circle, draw=black, double,fill=black!20] {\tiny $DGU_6$};
			\draw (s6)+(2,0) node(s7) [circle, draw=black, double,fill=black!20] {\tiny $DGU_7$};
			\draw (s6)+(1,-1.4) node(s8) [circle, draw=black, double,fill=black!20] {\tiny $DGU_8$};
			\draw (s8)+(2,0) node(s9) [circle, draw=black, double,fill=black!20] {\tiny $DGU_9$};
			
			\draw[blue] (s1) to(s2);
			\draw[blue] (s2) to(s3);
			\draw[blue] (s2) to(s4);
			\draw[blue] (s4) to(s5);
			\draw[blue] (s4) to(s6);
			\draw[blue] (s6) to(s7);
			\draw[blue] (s6) to(s8);
			\draw[blue] (s8) to(s9);
		\end{circuitikz}
		\caption{Schematic of the DCmG. Blue lines represent the physical interconnection between DGUs, realized as resistive power lines.
		}
		\label{fig:microgridnetwork}
	\end{subfigure}\vspace{.3cm}
	\begin{subfigure}{0.48\textwidth}
		\centering
		\ctikzset{bipoles/length=0.65cm}
		\tikzstyle{every node}=[font=\tiny]
		\begin{tikzpicture}[scale=0.5]
		\draw (1.5,4)
		to [short](1.5,4.5)
		to [short](3.5,4.5)
		to [short](3.5,0.5)
		to [short](1.5,0.5)
		to [short](1.5,4)
		to [short](1,4)
		to [battery, o-o](1,1)
		to [short](1.5,1)
		to [short](1,1);
		\node at (2.5,2.5){\footnotesize \textbf{Buck $i$}};
		\draw[-latex] (4,1.25) -- (4,3.75)node[midway,right]{$V_{ti}$};
		\draw (3.5,4) to [short](4,4)
		to [short](4.5,4)
		to [R=$R_{t}$] (6,4)
		to [L=$L_{t}$] (7.5,4)
		to [short, i=$\textcolor{green}{I _{ti}}$, -] (8.5,4)
		to [short](9,4)
		to [C, a=$C_{t}$, -] (9,1)
		to [short](4,1)
		to [short](3.5,1);
		\draw[blue] (12.3,4)  to [R=\textcolor{blue}{$R_{ij}$},color=blue,draw=blue] (15.8,4) 
		to [short,color=blue, -o] (16,4) node[anchor=north,above]{$V_j$};
		\draw (8.5,4) to (11,4)
		to [R ,a=$R_L$] (11 ,1)
		to [short] (9,1);
		\draw[blue] (11,1)
		to [short,color=blue, -o] (16,1);
		\draw (11,4) to [short](11.5,4);
		\draw (11,4) node[anchor=north, above]{$\textcolor{red}{V_i}$}  to [short](11,2.9);
		\draw[blue] (11,4) to [short](12.5,4);
		\node at (11,4.6)[anchor=north, above]{$PCC_i$} ;
		\draw[black, dashed] (.5,.25) -- (12.45,.25) -- (12.45,5.5) -- (.5,5.5)node[sloped, midway, above]{{\footnotesize \textbf{DGU and Load $i$ }}}  -- (.5,.25);
		\draw[black, dashed] (12.7,.25) -- (15.5,.25) -- (15.5,5.5) -- (12.7,5.5)node[sloped, midway, above]{{\footnotesize \textbf{Power line $ij$}}}  -- (12.7,.25);
		\draw[red,o-] (10.9,4.15) -- (11.7,3) to (11.7,-0.8);
		\draw[red,latex-](8,-0.8)-- (9,-0.8) --  (10,0)-- (11.7,0);
		\draw[green,o-latex] (8.5,4.15) to (8.5,1.75) -- (8.5,-0.3) to (8,-0.3);
		\draw (10,-1.3) node(a) [black, draw,fill=white!20] {$\normalsize{\int}$};
		\draw[-latex] (a.west) to (8,-1.3);
		\draw (11.7,-1.3) node(b)[ circle, draw=black, minimum size=12pt, fill=lightgray!20]{};
		\draw[red, -latex] (11.7,1.75)  -- (b.north) node[pos=0.9, left]{\textcolor{black}{\normalsize{-}}};
		\draw[latex-] (a.east) -- (b.west);
		\draw[-latex] (12.8,-1.3)  -- (b.east) node[pos=0.7,above]{{+}} node[pos=0.25,right]{$V_{ref,i}$};
		\draw[fill=lightgray] (8,-1.55) -- (8,-0.05) -- (6.5,-0.8) -- (8,-1.55);
		\node at (7.5,-0.8) {$K_{pr}$};
		\draw[-latex] (6.5,-0.8) -- (5.5,-0.8) -- (5.5,1.5) -- (5.5,2.5) -- (5,2.5);
		\end{tikzpicture}
		\caption{Electrical scheme of $i^{th}$ DGU along with load, connecting line(s), and local primary voltage regulator. }
		\label{fig:DGUi}
	\end{subfigure}
	\caption{DCmG considered in Section~\ref{subsec:SimulationResults_Vector}.}
	\label{fig:microgrid}\vspace{-0.6cm}
\end{figure}
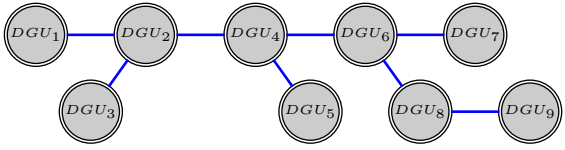
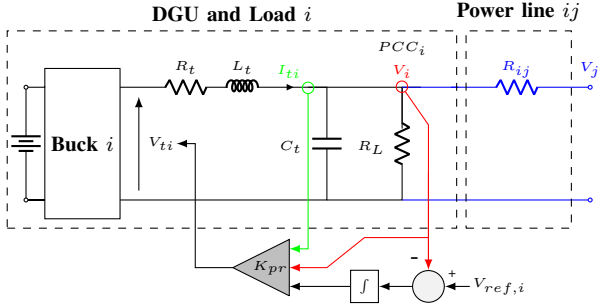

In this section, we consider a DCmG consisting of identical DGUs that are interconnected as shown in Figure~\ref{fig:microgridnetwork} and study the consensusability problem. Each DGU is modeled as a Buck converter with an $RLC$ filter, connecting a voltage source to a resistive load and neighboring DGUs, as shown in Figure~\ref{fig:DGUi}. Each DGU also includes a primary voltage controller for steering the voltage value $V_i$ at the point of common coupling (PCC) to the reference value $V_{ref,i}$. As depicted in Figure~\ref{fig:DGUi}, these primary controllers, developed in~\cite{nahata2020passivity}, have a static state-feedback structure captured by the gain $K_{pr}=[k_{pr,1},k_{pr,2},k_{pr,3}]$ and integral action. On applying Kirchoff's voltage and current laws, the dynamics of the DGU $i$ is written as\vspace{-0.1cm}
\begin{equation}\label{eq:DGUdynamics}
\begin{aligned}
C_{t}\dot{V}_{i} &= -\frac{V_i}{R_L}+I_{ti} -\sum_{j\in\Ncal_i^p}\frac{1}{R_{ij}}(V_i-V_j)\\
L_{t}\dot{I}_{ti} &= (k_{pr,1}-1)V_{i}+(k_{pr,2}-{R_{t}})I_{ti}+k_{pr,3}v_i\\
\dot{v}_i &= -V_i+V_{ref,i}
\end{aligned},
\end{equation}
where $I_{ti}$ is the filter current passing through the inductance and $v_i$ is the integrator state of the primary controller. In this section, we seek to develop a secondary controller modifying the voltage references $V_{ref,i}$ of each DGU to achieve consensus on the states of all DGUs in the DCmG. Specifically, we consider a networked controller given by \vspace{-0.1cm}
\begin{equation}\label{eq:Vrefi}
V_{ref,i} = V_{ref}+K\sum_{j\in\Ncal_i^c}b_{ij}(x_i-x_j),
\end{equation}\vspace{-0.4cm}

\noindent where $x_i \triangleq [V_i,I_{ti},v_i]^\top$ is the state of DGU $i$ and $V_{ref}=48V$ is a common nominal voltage reference for all DGUs.
We note that DGUs in the DCmG are physically coupled to each other and the secondary controller~\eqref{eq:Vrefi} is based on a communication network (inducing the set of neighbors $\Ncal_i^c$). Therefore, after discretizing \eqref{eq:DGUdynamics} with forward Euler method with a sampling period of $T_s=0.1ms$, the overall dynamics of the DCmG can be written as in \eqref{eq:LIMAS_dyn}, where the system matrices are defined as $ A = \I_3-T_sA_{ct}$, $B = \begin{bmatrix}
0 & 0 & T_s
\end{bmatrix}^\top,$\vspace{-0.2cm}
\begin{equation*}
	\begin{aligned}
		A_{ct} = \begin{bmatrix}
			-\frac{1}{R_LC_t} & \frac{1}{Ct}& 0 \\
			\frac{k_{pr,1}-1}{L_t} & \frac{k_{pr,2}-R_t}{L_t} & \frac{k_{pr,3}}{L_t} \\
			-1 & 0 & 0
		\end{bmatrix},
		\enskip &A_p = \begin{bmatrix}
			\frac{T_s}{C_t} & 0 & 0 \\
			0 & 0 & 0 \\
			0 & 0 & 0
		\end{bmatrix}.
	\end{aligned}
\end{equation*}

	\begin{figure}[t]
	\centering
	\includegraphics[width=0.45\textwidth]{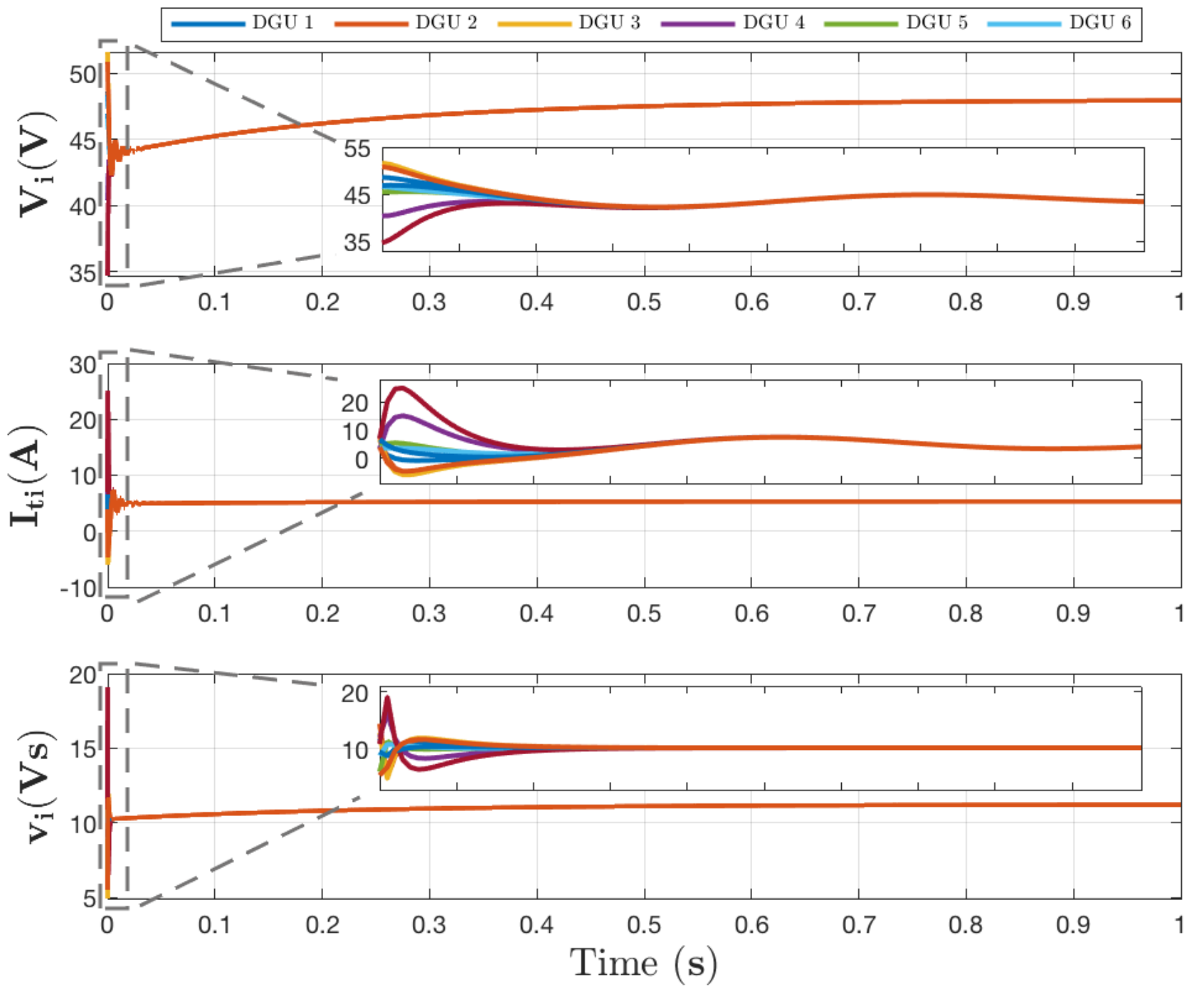}	\vspace{-.2cm}
	\caption{States of the DGUs in the DCmG equipped with consensus controllers. The zoomed-in figures show that consensus is achieved very quickly.}\vspace{-0.6cm}
	\label{fig:microgrid_states}
\end{figure}

We use the parameter values $R_t = 0.2$~$\Omega$, $C_t = 2.2$~$mF$, $L_t = 1.8$~$mH$, $R_L=9$~$\Omega$, and $K_{pr} = [-2.13,-0.16,13.55]$ taken from~\cite{nahata2020passivity}. Moreover, $\Lcal_p$ is derived from the physical interconnection topology in Figure~\ref{fig:microgridnetwork} and edge weights $1/R_{ij},$ $\forall i \in \{1,\dots,9\}$, $j\in\Ncal_i^p$, where power line resistances $R_{ij}$ are selected randomly in the interval $[4,8]\Omega$. In order to verify Assumption~\ref{ass:lapl_commute}, we assume a complete communication graph with uniform edge weights. Therefore, $\Lcal_c = \I_9-\frac{1}{9}\1_9\1_9^\top$. With these definitions, it is easy to verify that Assumption~\ref{ass:controllability} also holds. Hence, Corollary~\ref{crl:consensusability_numerical} can be utilized.

We construct the necessary matrices as shown in \eqref{eq:LP_suff_prop}-\eqref{eq:V_H_definitions}, \eqref{eq:Vdagger}-\eqref{eq:Psi} and verify that the LP in \eqref{eq:LP_suff_prop_red} is feasible. Hence, a consensus-enabling controller gain $K$ can be designed. With this controller in place, we run a simulation from random initial conditions of DGUs, based on their continuous-time dynamics in \eqref{eq:DGUdynamics}. Figure~\ref{fig:microgrid_states} shows that consensus is quickly reached as Corollary~\ref{crl:consensusability_numerical} certifies, and the voltages are regulated, albeit relatively slowly, towards the nominal reference value of $48V$. Furthermore, condition N1 is satisfied, as Theorem~\ref{prop:consensusability_necessary} guarantees for this consensusable LIMAS.

In order to show the effect of physical coupling on consensusability, we gradually increase $\Delta_p$ by scaling down the line resistances, i.e., we use $\tilde{R}_{ij}\triangleq \xi R_{ij}$ in the definition of $\Lcal_p$ for $\xi\in(0,1)$. For $\xi=0.072$, the LP in~\eqref{eq:LP_suff_prop_red} becomes infeasible. Similarly, if a DCmG is not consensusable using the proposed results, weakening its physical coupling could make it consensusable.

We next investigate whether our consensusability test in Corollary~\ref{crl:consensusability_numerical} can be used even when Assumption~\ref{ass:lapl_commute} is not satisfied. We do this by changing the topology of $\Gcal_c$. We see that the LP in~\eqref{eq:LP_suff_prop_red} is infeasible for $\mathcal{G}_c$ with circle and star topologies and unit edge weights. Moreover, although~\eqref{eq:LP_suff_prop_red} is feasible for a complete $\Gcal_c$, it can become infeasible upon removal of $2$ edges. By keeping the complete topology of $\Gcal_c$ and selecting non-uniform edge weights between $0$ and $1$, we obtain that~\eqref{eq:LP_suff_prop_red} becomes infeasible as well. This study reveals that our results depend critically on the satisfaction of Assumption~\ref{ass:lapl_commute}.

\vspace{-0.2cm}

\section{Conclusions and Future Perspectives}\label{sec:ConclusionAndFuturePerspectives}

In this paper, we considered linear MASs with physical interconnections among subsystems (LIMASs) and studied their consensusability properties.
We show that checking the consensusability of LIMASs is equivalent to a simultaneous stabilization problem and present a linear program (LP) based method for verifying simultaneous stabilization.
Based on this result, we present a numerical sufficient condition for the consensusability of a LIMAS.
Moreover, we propose several algebraic consensusability conditions that are either sufficient or necessary.
The derived results show that \textit{weak} physical coupling and \textit{densely-connected} physical and communication graphs are favorable for consensusability.

The presented results in the case of vector dynamics rely on the assumption that the Laplacians of physical and communication graphs commute.
Eliminating this assumption will be considered in future works.
Another direction for follow-up research is to study output synchronization for LIMASs.\vspace{-0.3cm}

\bibliographystyle{IEEEtran}
\bibliography{consensus_interconnected_MAS_TR}
	
\end{document}